\newif\ifjournal
\journalfalse

\ifjournal
\documentclass{jgaa-art}
\else
\documentclass{jgaa-art-modified-for-arxiv}
\fi

\usepackage[T1]{fontenc} 
\usepackage{graphicx,wrapfig}
\usepackage{amsmath}
\usepackage{amssymb}

\usepackage{booktabs}
\usepackage{multirow}
\usepackage{xspace}
\usepackage{tikz}
\usepackage[inline]{enumitem}
\usepackage[font=small,labelfont=bf]{caption}
\usepackage[subrefformat=parens,labelfont=default,font=small]{subcaption}

\usepackage{todonotes}

\ifjournal
\usepackage{lineno}
\linenumbers
\fi

\graphicspath{{figures/}}

\newcommand{\stic}{\textsf{STICK}}
\newcommand{\stick}{\stic\xspace}
\newcommand{\stickA}{\stic$_\mathsf{A}$\xspace}
\newcommand{\stickAB}{\stic$_\mathsf{AB}$\xspace}
\newcommand{\stickfix}{\stic$^\mathsf{fix}$\xspace}
\newcommand{\stickfixA}{\stic$^\mathsf{fix}_\mathsf{A}$\xspace}
\newcommand{\stickfixAB}{\stic$^\mathsf{fix}_\mathsf{AB}$\xspace}
\newcommand{\threePartition}{\textsf{3-PARTITION}\xspace}
\newcommand{\monThreeSat}{\textsf{MONOTONE-3-SAT}\xspace}
\newcommand{\shortto}[1]{%
	\parbox{#1}{\tikz{\draw[->](0,0)--(#1,0);}}
}
\newcommand{\enter}[1]{{\ensuremath{#1}}}
\newcommand{\exit}[1]{{\ensuremath{#1\shortto{.15cm}}}}

\newtheorem{thm}{Theorem}[section]
\newtheorem{lem}[thm]{Lemma}
\newtheorem{cor}[thm]{Corollary}
\newtheorem{definition}[thm]{Definition}
\newcommand{\myparagraph}[1]{\paragraph{#1}}

\begin{document}
\doi{} 
\Issue{0}{0}{0}{0}{0} 
\HeadingAuthor{Chaplick et al.} 
\HeadingTitle{Recognizing Stick Graphs} 
\title{Recognizing Stick Graphs \\
  with and without Length Constraints} 

\Ack{S.C.\ and A.W.\ acknowledge support from DFG grants
  WO$\,$758/11-1 and WO$\,$758/9-1, resp.}

\author[first]{Steven~Chaplick}{orcid.org/0000-0003-3501-4608}
\author[first]{Philipp~Kindermann}{orcid.org/0000-0001-5764-7719}
\author[first]{Andre~L\"offler}{loeffler@informatik.uni-wuerzburg.de}
\author[first]{Florian~Thiele}{--}
\author[first]{Alexander~Wolff}{orcid.org/0000-0001-5872-718X}
\author[first]{Alexander~Zaft}{--}
\author[first]{Johannes~Zink}{orcid.org/0000-0002-7398-718X}

\affiliation[first]{Institut f\"ur Informatik, \\ Universit\"at W\"urzburg,
	W\"urzburg, Germany}

\submitted{Nov.~4, 2019}%
\reviewed{Jan.~21, 2020}%
\revised{Mar. 2020}%
\accepted{}%
\final{}%
\published{}%
\type{Special Issue of GD 2019}%
\editor{D.~Archambault and Cs.~T\'oth}%

\maketitle


\ifjournal
\else
\vspace{50 pt}
\fi

\begin{abstract}
	Stick graphs are intersection graphs of horizontal and vertical line
	segments that all touch a line of slope~$-1$ and lie above this
	line.  De Luca et al.~[GD'18] considered the
	recognition problem of stick graphs when no order is given (\stick),
	when the order of either one of the two sets is given (\stickA), and
	when the order of both sets is given (\stickAB).  They showed
	how to solve \stickAB efficiently.  \par In this paper, we improve
	the running time of their algorithm, and we solve \stickA efficiently.
	Further, we consider variants of these problems where the lengths of
	the sticks are given as input.  We show that these variants of
	\stick, \stickA, and \stickAB are all NP-complete.
	On the positive side, we give an efficient solution for \stickAB
	with fixed stick lengths if there are no isolated vertices.
\end{abstract}

\Body 


\section{Introduction}

For a given collection $\mathcal{S}$ of geometric objects, the \emph{intersection graph of $\mathcal{S}$} has $\mathcal{S}$ as its vertex set and an edge whenever $S \cap S' \neq \emptyset$, for $S, S' \in \mathcal{S}$. 
This paper concerns \emph{recognition} problems for classes of
intersection graphs of restricted geometric objects, i.e., determining
whether a given graph is an intersection graph of a family of
restricted sets of geometric objects.
An established (general) class of intersection graphs is that of \emph{segment graphs}, the intersection graphs of line segments in the plane\footnote{We follow the common convention that parallel segments do not intersect and each point in the plane belongs to at most two segments.}. 
For example, segment graphs are known to include planar
graphs~\cite{ChalopinG09}.
The recognition problem for segment graphs is $\exists\mathbb{R}$-complete~\cite{KratochvilM94,Matousek14}\footnote{Note that $\exists\mathbb{R}$ includes NP, see~\cite{Matousek14,Schaefer09} for background on the complexity class~$\exists\mathbb{R}$.}.
On the other hand, one of the simplest natural subclasses of segment
graphs is that of the \emph{permutation} graphs, the
intersection graphs of line segments where there are two parallel
lines such that each line segment has its two end points on these
parallel lines\footnote{i.e., we think of the sequence of end points
	on the ``bottom'' line as one permutation $\pi$ on the vertices and
	the sequence on the ``top'' line as another permutation $\pi'$, where
	$uv$ is an edge if and only if the order of $u$ and $v$ differs in
	$\pi$ and $\pi'$.}.
We say that the segments are \emph{grounded} on these two lines.
The recognition problem for permutation graphs can be solved in linear
time~\cite{KratschMMS06}.  \emph{Bipartite} permutation graphs have an
even simpler intersection representation~\cite{SenS94}:
they are the intersection graphs of unit-length vertical and
horizontal line segments which are again double-grounded (without loss
of generality, both lines have a slope of $-1$).
The simplicity of bipartite permutation graphs leads to a
simpler linear-time recognition algorithm~\cite{sbs-bpg-DAM87} than
that of permutation graphs.

Several recent articles~\cite{CabelloJ17,CardinalFMTV-JGAA18,ChaplickFHW18,ChaplickHOSU17} compare and study the geometric intersection graph classes occurring between the simple classes, such as bipartite permutation graphs, and the general classes, such a segment graphs. 
Cabello and Jej\v{c}i\v{c}~\cite{CabelloJ17} mention that studying
such classes with constraints on the sizes or lengths of the objects
is an interesting direction for future work (and such constraints
are the focus of our work). 
Note that similar length restrictions have been considered for other geometric intersection graphs such as interval graphs~\cite{KlavikOS19,KoblerK015,PeerS97}.

When the segments are not grounded, but still are only horizontal and vertical, the class is referred to as \emph{grid intersection graphs} and it also has a rich history, see, e.g.,~\cite{ChaplickFHW18,ChaplickHOSU17,HartmanNZ91,Kratochvil94}. 
In particular, note that the recognition problem is NP-complete for
grid intersection graphs~\cite{Kratochvil94}.
But, if both the permutation of the vertical segments and the permutation of the horizontal segments are given, then the problem becomes a trivial check on the bipartite adjacency matrix~\cite{Kratochvil94}. 
However, for the variant where only one such permutation, e.g., the
order of the horizontal segments, is given, the complexity remains open. 
A few special cases of this problem have been solved efficiently~\cite{ChaplickDKMS14,dhklm-rdsg-TCS19,FelsnerMM13}, e.g., one such case~\cite{ChaplickDKMS14} is equivalent to the problem of \emph{level planarity testing} which can be solved in linear time~\cite{JungerLM98}.

In this paper we study recognition problems concerning so-called
\emph{stick} graphs, the intersection graphs of grounded vertical and
horizontal line segments (i.e., grounded grid intersection graphs).
Classes closely related to stick graphs appear in several application contexts, e.g., in \emph{nano PLA-design}~\cite{ShresthaTTU11} and detecting \emph{loss of heterozygosity events in the human genome}~\cite{CatanzaroCFHHHS17,HalldorssonATI11}.
Note that, similar to the general case of segment graphs, it was recently shown that the recognition problem for grounded segments (where arbitrary slopes are allowed) is $\exists\mathbb{R}$-complete~\cite{CardinalFMTV-JGAA18}.
So, it seems likely that the recognition problem for stick graphs is NP-complete (similar to grid intersection graphs), but thus far it remains open. 
The primary prior work on recognizing stick graphs is due to De Luca et al.~\cite{dhklm-rdsg-TCS19}.
They introduced two constrained cases of the stick graph recognition problem called \stickA and \stickAB, which we define next.
Similarly to Kratochv\'{i}l's approach to grid intersection graphs~\cite{Kratochvil94},
De~Luca et al.\ characterized 
stick graphs through their bipartite adjacency matrix and used this
result as a basis to develop a polynomial-time algorithm to
solve \stickAB.

\begin{definition}[\stick]
	Let~$G$ be a bipartite graph with vertex set $A \dot\cup B$,
	and let~$\ell$ be a line with slope $-1$.  
	Decide whether~$G$ has an intersection representation where 
	the vertices in~$A$ are vertical line segments whose bottom end-points lie~on~$\ell$ and 
	the vertices in~$B$ are horizontal line segments whose left end-points lie~on~$\ell$.%
	\footnote{Note that De Luca et al.~\cite{dhklm-rdsg-TCS19}
		regarded~$A$ as the set of horizontal segments.}
	Such a representation is a \emph{stick representation} of~$G$, the line~$\ell$ is the \emph{ground~line},~the
	segments are called \emph{sticks}, and the point where a stick meets $\ell$ is its \emph{foot point}.
\end{definition}

\begin{definition}[\stickA/\stickAB]
	In the problem \stickA (\stickAB) we are given an instance of the \stick problem and additionally an order~$\sigma_A$ (orders~$\sigma_A, \sigma_B$) of the vertices in~$A$ (in~$A$ and~$B$).
	The task is to decide whether there is a stick representation that respects $\sigma_A$ ($\sigma_A$ and $\sigma_B$).
\end{definition}

\paragraph{Our Contribution.}
We first revisit the problems \stickA and \stickAB defined by De Luca
et al.~\cite{dhklm-rdsg-TCS19}. 
We provide the first efficient algorithm for \stickA and a faster algorithm for \stickAB;
see Section~\ref{sec:var_length}.   
For our \stickA algorithm we introduce a new tool, semi-ordered trees
(see Section~\ref{sec:semi-orderedtrees}), as a way to capture all
possible permutations of the horizontal sticks which occur in a
solution to the given \stickA instance.  We feel that this data
structure may be of independent interest.  Then we
investigate the direction suggested by Cabello and
Jej\v{c}i\v{c}~\cite{CabelloJ17} where specific lengths are given for
the segments of each vertex.  In particular, this can be thought of as
generalizing from unit stick graphs (i.e., bipartite permutation
graphs), where every segment has the same length.  While bipartite
permutation graphs can be recognized in linear
time~\cite{sbs-bpg-DAM87}, it turns out that all of the new problem
variants (which we call \stickfix, \stickfixA, and \stickfixAB) are
NP-complete; see Section~\ref{sec:fixed}.  Finally, we give an
efficient solution for \stickfixAB (that is, \stickAB with fixed stick
lengths) for the special case that there are no isolated vertices (see
Section~\ref{sec:fixed-length-a-b-given-no-isolated}).  We conclude
and state some open problems in Section~\ref{sec:open_problems}.
Our results are summarized in Table~\ref{table:res}.

\begin{table}[tb]
	\centering
	\caption{Previously known and new results for deciding whether
		a given bipartite graph $G=(A \dot\cup B, E)$ is a stick graph.
		In $O(\cdot)$, we dropped $|\cdot|$.  We abbreviate
                vertices by vtc., NP-complete by NPC, Theorem by~T.,
                and Corollary by~C.}
	\label{table:res}
	
	\begin{tabular}{@{}l@{\qquad}c@{~~}l@{~~}c@{~~}l@{\quad}c@{~~}l@{~~}c@{~~}l@{}}
		\toprule
		\multirow[b]{2.4}{4 pt}{given \\ order} & \multicolumn{4}{c}{variable length} &
		\multicolumn{4}{c}{fixed length} \\
		\cmidrule(l{0ex}r{3ex}){2-5}
		\cmidrule(l{0ex}r{1ex}){6-9}
		& \multicolumn{2}{c@{\quad}}{old} & \multicolumn{2}{c@{\quad}}{new} & \multicolumn{2}{c@{\quad}}{isolated vtc.} & \multicolumn{2}{c}{no isolated vtc.}\\
		\midrule
		$\emptyset$ & \multicolumn{2}{c@{~~}}{unknown} & \multicolumn{2}{c@{\quad}}{unknown} & NPC & [T.~\ref{thm:stick-fixedlengths}] & NPC & [T.~\ref{thm:stick-fixedlengths}]\\
		$A$ & \multicolumn{2}{c@{~~}}{unknown} & $O(AB)$ & [T.~\ref{thm:sticka}] & NPC & [T.~\ref{thm:sticka-fixedlengths}] & NPC & [T.~\ref{thm:sticka-fixedlengths}]\\
		$A$,$B$ & $O(AB)$ & \cite{dhklm-rdsg-TCS19} & $O(E)$ & [T.~\ref{thm:stickab}] & NPC & [C.~\ref{cor:stickab-fixedlengths-isolated}] & $O((A+B)^2)$& [C.~\ref{cor:stickab-fixed-noisolated}]\\
		\bottomrule
	\end{tabular}
\end{table}

\section{Sticks of Variable Lengths}
\label{sec:var_length}

In this section, we provide algorithms that solve the \stickAB problem
in $O(|A|+|B|+|E|)$ time (see Section~\ref{sec:stickAB},
Theorem~\ref{thm:stickab}) and the \stickA problem in $O(|A|\cdot
|B|)$ time (see Section~\ref{sec:stickA}, Theorem~\ref{thm:sticka}).
Between these subsections, in Section~\ref{sec:semi-orderedtrees}, we
describe \emph{semi-ordered trees}, an essential tool reminiscent of
PQ-trees that we will use for the latter algorithm.
This tool will allow us to express the different ways one can order the horizontal segments for a given instance of \stickA. 

\subsection{Solving \stickAB in $O(|A|+|B|+|E|)$ time}
\label{sec:stickAB}

De Luca et al.~\cite{dhklm-rdsg-TCS19} showed how to compute, for a
given graph $G=(A \cup B, E)$ and orders $\sigma_A$ and $\sigma_B$,
a \stickAB representation in $O(|A| \cdot |B|)$ time (if such a
representation exists).  We improve upon their result in this section.
Namely, we prove the following theorem.

\begin{thm}\label{thm:stickab}
	\stickAB can be solved in $O(|A|+|B|+|E|)$ time.
\end{thm}
\begin{proof}
  We apply a sweep-line approach (with a vertical
  sweep-line moving rightwards) where each vertical stick~$a_i\in A$
  corresponds to two events: the \emph{enter event of $a_i$}
  (abbreviated by $\enter{i}$) and the \emph{exit event of $a_i$}
  (abbreviated by~$\exit{i})$.

	Let $\sigma_A=(a_1,\ldots,a_{|A|})$ and $\sigma_B=(b_1,\ldots,b_{|B|})$.
	Let $\beta_i$ denote the largest index such that~$b_{\beta_i}$ has a neighbor in $a_1,\ldots,a_i$.
	Let $\hat{B}^{\enter{i}}$ be the subsequence of
        $(b_1,\ldots,b_{\beta_i})$ of those vertices that have a
        neighbor in $a_{i},\ldots,a_{|A|}$, and let
        $\hat{B}^{\exit{i}}$ be the subsequence of
        $(b_1,\ldots,b_{\beta_i})$ of those vertices that have a
        neighbor in $a_{i+1},\ldots,a_{|A|}$.
	At every event $p\in\{\enter{i},\exit{i}\}$, we maintain the
        following invariants.
	\begin{enumerate}[label=(\roman*)]
        \item\label{item-stickAB-rep} We have a valid representation
          of the subgraph of~$G$ induced by the vertices
          $b_1,\ldots,b_{\beta_i},a_1,\ldots,a_i$.
        \item\label{item-stickAB-coord} The x-coordinates of the foot
          points of $\{b_1,\ldots,b_{\beta_i},a_1,\ldots,a_i\}$ are
          unique integers in the range from~1 to $\beta_i+i$.
        \item\label{item-stickAB-length} For the vertices in
          $\{b_1,\ldots,b_{\beta_i},a_1,\ldots,a_i\} \setminus
          \hat{B}^p$, both endpoints are set.
	\end{enumerate} 
    We initialize $\hat{B}^{\exit{0}} = \hat{B}^{\enter{0}} = \emptyset$ and $\beta_0 = 0$. 
    Here, our invariants trivially hold. Now suppose $i \geq 1$. 
    In the following, we don't create a new variable $\hat{B}^p$ for each event~$p$,
    but we update a single variable $\hat{B}$, viewing~$\hat{B}^p$
    as the state of $\hat{B}$ during event~$p$.
	
	Consider the enter event of $a_i$.  We set the x-coordinate
        of~$a_i$ to $\beta_i+i$.  We place the foot points of vertices
        $b_{\beta_{i-1}+1},\ldots, b_{\beta_i}$ (if they exist)
        between $a_{i-1}$ and $a_i$ in this order and create
        $\hat{B}^{\enter{i}}$ by appending them to
        $\hat{B}^{\exit{(i-1)}}$ in this order.
        All neighbors of~$a_i$ have
        to start before~$a_i$, and they have to be a suffix of
        $\hat{B}^{\enter{i}}$.  
        If this is not the case, then we simply reject as this is
        a negative instance of the problem.         
        This is easily checked in $O(\deg(a_i))$ time.  
        The upper endpoint of~$a_i$ is placed 1/2 a unit above
        the foot point of its first neighbor in this suffix.
	As such, the invariants \ref{item-stickAB-rep}--\ref{item-stickAB-length} are maintained.
	
	Consider the exit event of~$a_i$.  For each neighbor~$b_j$
        of~$a_i$, we check whether~$a_i$ is the last neighbor of~$b_j$
        in~$\sigma_A$.  In this case, we finish~$b_j$ and set the
        x-coordinate of its right endpoint to~$\beta_i+i+1/2$.
	Now $\hat{B}^{\exit{i}}$ consists of all vertices
        in~$\hat{B}^{\enter{i}}$ except those that we just finished.
	This again maintains invariants
        \ref{item-stickAB-rep}--\ref{item-stickAB-length}.
        Note that processing the exit event always succeeds, i.e., negative instances are detected purely in the enter events. 

	Hence, if we reach and complete the exit event of~$a_{|A|}$,
        we obtain a \stickAB representation of~$G$.
        Otherwise,~$G$ has no such representation.
	Clearly, the whole algorithm runs in $O(|A|+|B|+|E|)$ time.

	Note that, even though we have not explicitly discussed
        isolated vertices, these can be easily realized by sticks of
        length~1/2.
\end{proof}

\subsection{Data Structure: Semi-Ordered Trees}
\label{sec:semi-orderedtrees}

In the \stickA problem, the goal is to find a permutation of the horizontal sticks~$B$
that is consistent with the fixed permutation of the vertical sticks~$A$.
To this end, we will make use of a data structure that allows us to capture
many permutations subject to consecutivity constraints. This might remind
the reader of other similar but distinct data structures such as PQ-trees~\cite{BoothL76}.

An \emph{ordered tree} is a rooted tree where the order of the children
around each internal is specified. The permutation \emph{expressed} by an ordered tree~$T$
is the permutation of its leaves in the pre-order traversal of~$T$.

Generalizing this, we define a
\emph{semi-ordered tree} where, for each node, there is a fixed permutation
for a subset of the children and the remaining children are free.
Namely, for each node~$v$, we have
\begin{enumerate}[label=(\roman*)]
	\item a set $U_v$ of \emph{unordered} children,
  \item a set $O_v$ of \emph{ordered} children, and
  \item a fixed permutation $\pi_v$ of $O_v$; see Fig.~\ref{fig:stickA-semiorderedtree}.
\end{enumerate}
\begin{figure}
  \begin{subfigure}[b]{.32\textwidth}
    \centering
    \includegraphics[page=1]{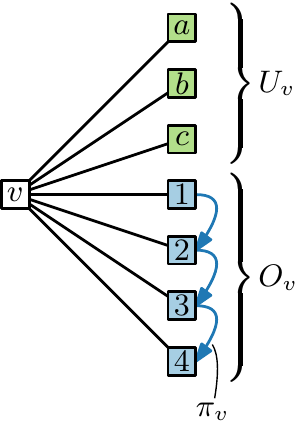}
  \end{subfigure}
  \hfill
  \begin{subfigure}[b]{.24\textwidth}
    \centering
    \includegraphics[page=2]{stickA-definitions}
  \end{subfigure}
  \hfill
  \begin{subfigure}[b]{.32\textwidth}
    \centering
    \includegraphics[page=3]{stickA-definitions}
  \end{subfigure}

  \begin{subfigure}[t]{.32\textwidth}
    \caption{a semi-ordered tree~$S$}
  \end{subfigure}
  \hfill
  \begin{subfigure}[t]{.24\textwidth}
    \caption{an ordered tree obtained from~$S$}
  \end{subfigure}
  \hfill
  \begin{subfigure}[t]{.32\textwidth}
    \caption{an ordered tree that cannot be obtained from~$S$}
  \end{subfigure}

  \caption{Definition of a semi-ordered tree.}
  \label{fig:stickA-semiorderedtree}
\end{figure}
Hence, every node (except the root) is ordered or unordered depending on its parent.
We \emph{obtain} an ordered tree from a semi-ordered tree by fixing, for each node~$v$,
a permutation $\pi'_v$ of $O_v\cup U_v$ that contains~$\pi_v$ as a
subsequence. In this way, a permutation $\pi$ is \emph{expressed} by a semi-ordered tree~$S$
if there exists an ordered tree~$T$ that expresses~$\pi$ and can be obtained from~$S$.

\subsection{Solving \stickA in $O(|A|\cdot|B|)$ time}
\label{sec:stickA}

	Let $G=(A \dot\cup B, E)$ and $\sigma_A=(a_1,\ldots,a_{|A|})$ be the
  input.
  We assume that $G$ is connected and discuss otherwise at the end of this section.
  
  As in the algorithm for \stickAB, we apply a sweep-line approach (with a vertical
  sweep-line moving rightwards) where each vertical stick~$a_i\in A$
  corresponds to two events: the \emph{enter event of $a_i$}
  (abbreviated by $\enter{i}$) and the \emph{exit event of $a_i$}
  (abbreviated by~$\exit{i})$.
    
  \myparagraph{Overview.}
  Informally, for each event~$p\in\{\enter{i},\exit{i}\}$, we will maintain all representations
  of the subgraph seen so far subject to certain horizontal sticks
  continuing further (those that intersect the sweep-line and some vertical stick before it).
  We denote by $G^p$ the induced subgraph of $G$ containing $a_1,\ldots,a_i$ and their neighbors.
  We distinguish the neighbors as those that are \emph{dead} (that is, have all neighbors before the sweep-line)
  and those that are \emph{active} (that is, have neighbors before and after the sweep-line).
  Namely, 
  \begin{itemize}
    \item $B^{p}$ consists of all sticks of~$B$ in $G^p$;
    \item $D^{\enter{i}}$ consists of all (dead) sticks of~$B^{\enter{i}}$ with no neighbor in $a_i,\ldots,a_{|A|}$; and
    \item $D^{\exit{i}}$ consists of all (dead) sticks of~$B^{\exit{i}}$ with no neighbor in $a_{i+1},\ldots,a_{|A|}$. 
  \end{itemize}
  
  To this end, we maintain an ordered forest $\mathcal{T}^p$ of
  special semi-ordered trees that encodes all realizable permutations
  (defined below) of the set of horizontal sticks~$B^p$ as the
  permutations expressed by $\mathcal{T}^p$; see
  Fig.~\ref{fig:stickA-example}.
	A permutation~$\pi$ of $B^p$ is \emph{realizable} if there is a stick representation of the graph~$H^p$
	obtained from~$G^p$ by adding a vertical stick to the right of $a_i$ neighboring all horizontal sticks in $B^p$ where $B^p$ is drawn top-to-bottom in order~$\pi$.

  In the enter event of $a_i$, $B^i$ comprises $B^\exit{i-1}$ and all
  vertices of~$B$ that neighbor~$a_i$ and aren't in the data structure
  yet (we call these \emph{entering vertices}).  We constrain the data
  structure so that all the neighbors of~$a_i$ must occur after
  (below) the non-neighbors of~$a_i$.  The set~$D^p$ of dead vertices
  remains unchanged with respect to the last exit event, that is, $D^i
  = D^\exit{(i-1)}$.

  In the exit event of $a_i$, $D^\exit{i}$ comprises $D^i$ and all
  sticks of~$B^i$ that do not have any neighbor~$a_j$ with $j>i$,
  i.e., those having~$a_i$ as their last neighbor (we call these
  \emph{exiting vertices}).  No new horizontal sticks appear in an
  exit event, hence $B^\exit{i}=B^i$.
  
  \begin{figure}
    \centering
    \subcaptionbox{the graph $G$}{\includegraphics[page=2]{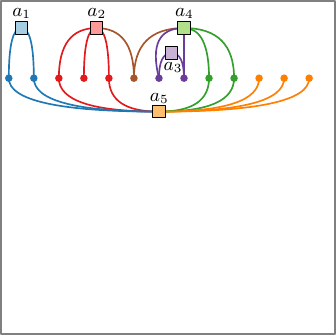}}
    \hfil
    \subcaptionbox{enter event of~$a_4$}{\includegraphics[page=5]{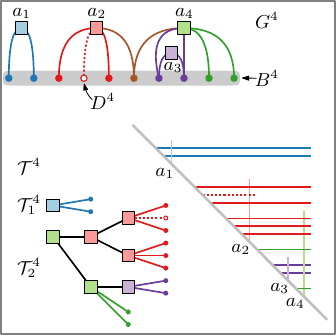}}
    
    \bigskip
    
    \subcaptionbox{exit event of~$a_4$}{\includegraphics[page=8]{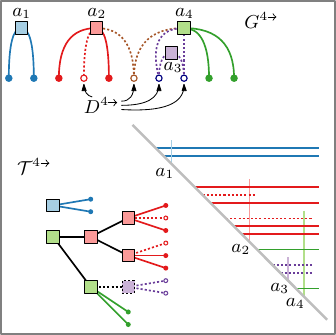}}
    \hfil
    \subcaptionbox{enter event of~$a_5$}{\includegraphics[page=11]{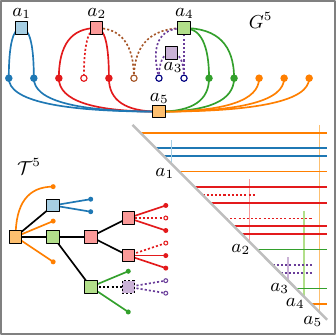}}
    \caption{An example for the data structures and a drawing
      corresponding to a realizable permutation of~$B^p$ for a
      graph~$G$ with $\pi_a=(a_1,\ldots,a_k)$.  Vertices in~$A$ are
      drawn as squares; vertices in~$B$ are drawn as circles.  The
      vertices and edges are colored by the entering event in which
      they appear.  Dead vertices are drawn as empty circles; their
      edges are dotted.  In~(d), the leaves are permuted in the order
      in which they are drawn.}
    \label{fig:stickA-example}
  \end{figure}
	
	\myparagraph{Data structure.}
	See Fig.~\ref{fig:stickA-example} for an example.
	Consider any event~$p$. Observe that~$G^p$ may consist of several connected components $G^p_1,\ldots,G^p_{k_p}$. 
	The components of~$G^p$ are naturally ordered from left to right by $\sigma_A$.
	Let $B_j^p$ denote the vertices of $B^p$ in~$G^p_j$.
	In this case, in every realizable permutation of $B^p$, the vertices of $B_j^p$ must come before the vertices of~$B_{j+1}^p$. 
	Furthermore, the vertices that will be introduced any time later can only be placed
	at the beginning, end, or between the components.
	Hence, to compactly encode the realizable permutations, it suffices to do so for each component~$G^p_j$ individually via a semi-ordered tree $T^p_j$. 
	Namely, our data structure will be $\mathcal{T}^p=(T^p_1,\ldots, T^p_{k_p})$.
		Each data structure $T^p_j$ is a special semi-ordered tree in which the leaves correspond to the vertices of $B_j^p$, all leaves are unordered, and all internal vertices are ordered.
	
  \myparagraph{Correctness and event processing.}
  We argue by induction that this data structure is sufficient to
  express the realizable permutations of~$B^p$.
  We maintain the following invariants for each event~$p$ during the
  execution of the algorithm.
  \begin{enumerate}[label=(I\arabic*)]
    \item The set of permutations expressed by $\mathcal{T}^p$ contains all 
      permutations of~$B^p$ which occur in a stick representation of~$G$.
    \item The set of permutations expressed by $\mathcal{T}^p$ contains only
      permutations of~$B^p$ which occur in a stick representation of~$G^p$.
  \end{enumerate}
  Since $G^{\exit{|A|}}=G$ and $B^{\exit{|A|}}=B$, after the final step these invariants ensure that
  our data structure expresses exactly those permutations of $B$ which occur in 
  a stick representation of~$G$.
  
  Recall that our data structure consists of an ordered set of semi-ordered trees.
  Note that these invariants also apply to each semi-ordered tree individually, that is, to its corresponding connected component.
  
	In the base case, consider the enter event of $a_1$. 
	Our data structure consists of a single component~$G_1^{\enter{1}}$ and clearly a single node with a leaf-child for every neighbor of~$a_1$ captures all possible permutations.
	
	In the exit event of $a_i$, we do not change the shape of $\mathcal T^{\enter{i}}$, that is, $\mathcal T^{\exit{i}}=\mathcal T^{\enter{i}}$.
  Then, in $\mathcal T^{\exit{i}}$, we mark the exiting vertices as dead and add them to $D^{\exit{i}}$.
  We further mark any internal node in $\mathcal T^{\exit{i}}$ that contains only dead leaves in its subtree as dead as well.
	Obviously, this procedure maintains all the invariants.
	
	Now consider the enter event of~$a_i$ and assume that 
	we have the data structure $\mathcal{T}^{\exit{(i-1)}}=(T_1^{\exit{(i-1)}},\ldots, T_{k_{i-1}}^{\exit{(i-1)}})$.
	The essential observation is that the neighbors of~$a_i$ must form a suffix
	of the active vertices in~$B^{\exit{(i-1)}}$ in
	every realizable permutation after the enter event, which we will enforce in the following. 
	Namely, either 
	\begin{itemize}
		\item all active vertices in~$B^\exit{(i-1)}$ are adjacent to~$a_i$,
		\item none of them are adjacent to~$a_i$, or 
		\item there is an $s \in \{1,\dots,k_{i-1}\}$ such that 
		\begin{enumerate*}[label=(\roman*),nosep]
			\item $B_{s}^{\exit{(i-1)}}$ contains at least one neighbor of $a_i$; 
			\item all active vertices in $B_{s+1}^{\exit{(i-1)}},\ldots, B_{k_{i-1}}^{\exit{(i-1)}}$ are neighbors of $a_i$; and 
			\item no active vertices in $B_{1}^{\exit{(i-1)}},\ldots, B_{s-1}^{\exit{(i-1)}}$ are adjacent to $a_i$;
			see Fig.~\ref{fig:stickA-construction-exit}. 
		\end{enumerate*}
	\end{itemize}
	Otherwise, there is no realizable permutation for this event and consequently for~$G$. 
	The first two cases can be seen as degenerate cases (with $s=0$ or $s=k_{i-1}+1$) of the general case below.
	
	\begin{figure}[t]
		\centering
		\subcaptionbox{$\mathcal{T}^\exit{(i-1)}$\label{fig:stickA-construction-exit}}{\includegraphics[page=1]{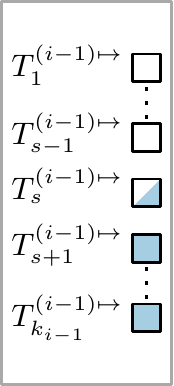}}
		\hfill
		\subcaptionbox{transformation of $T$\label{fig:stickA-construction-T}}{\includegraphics[page=2]{stickA-construction}}
		\hfill
		\subcaptionbox{$\mathcal{T}^\enter{i}$\label{fig:stickA-construction-enter}}{\includegraphics[page=3]{stickA-construction}}
		\caption{Construction of $\mathcal{T}^{\enter{i}}$. 
                  Leaves are drawn as small circles.  The leaves at
                  the new node~$z$ are the entering vertices.  Only
                  active vertices are shown.}
		\label{fig:stickA-construction}
	\end{figure}
	
	We first show how to process $T_s^{\exit{(i-1)}}$; see Fig.~\ref{fig:stickA-construction-T}.
	Afterwards we will create the data structure $\mathcal{T}^{\enter{i}}$.
	We create a tree~$T$ that expresses precisely the subset of the permutations expressed by $T_s^{\exit{(i-1)}}$ where all leaves that are neighbors of~$a_i$ occur as a suffix.
	We initialize $T=T_s^{\exit{(i-1)}}$.
	If all active vertices in $B_s^{\exit{(i-1)}}$ are neighbors of~$a_i$, then we are already done.
	
	Otherwise, we say that a node of~$T$ is \emph{marked} if all
        active leaves in its subtree are neighbors of~$a_i$; it is
        \emph{unmarked} if no active leaf in its subtree is a neighbor
        of~$a_i$; and it is \emph{half-marked} otherwise.  Note that
        the root of~$T$ is half-marked.

	Since the neighbors of~$a_i$ must form a suffix of the active leaves, the marked non-leaf children of a half-marked node form a suffix of the active children,
	the unmarked non-leaf children form a prefix of the active children, and there is at most one half-marked child. 
	Hence, the half-marked nodes form a path in~$T$ that starts in the root; 
	otherwise, there are no realizable permutations for this event and subsequently for $G$. 
	
	We traverse the path starting from the deepest descendant and
        ending at the root, i.e., bottom-to-top.  Let~$x$
        be a half-marked node, and let~$y$ be its half-marked child
        (if it exists); see Fig.~\ref{fig:stickA-construction-T}.
	We have to enforce that in any ordered tree obtained from~$T$,
        the unmarked children of~$x$
	occur before~$y$ and the marked children of~$x$ occur after~$y$.
	We create a new marked vertex~$x'$.  This vertex receives
        the following children from~$x$: the marked leaf-children
        and the suffix of the non-leaf children starting after~$y$.
	Symmetrically, we create a new unmarked vertex~$x''$,
        which receives the following children from~$x$:
	the unmarked leaf-children 
        and the prefix of the non-leaf children ending before~$y$.
	Then we make~$x'$ and~$x''$ children of~$x$ such that
        $x''$ is before~$y$ and $y$ is before~$x'$.
	If this results in any internal node having no leaf-children
        and only one child, we merge this node with its parent.
        (Note that this can only happen to~$x'$ or~$x''$.)
	This ensures that for every permutation expressed by~$T$, the
        the subsequence of active vertices has the neighbors of~$a_i$ as a suffix.
  
  Note that every non-leaf of $T_s^{\exit{(i-1)}}$ is also a non-leaf in~$T$ with the same set of leaves in its subtree.
	In the pre-order traversal of any ordered tree obtained from~$T$,
  the non-leaves of~$T_s^{\exit{(i-1)}}$ are visited in the same order
  as in the pre-order traversal of any ordered tree obtained from~$T_s^{\exit{(i-1)}}$.
  This implies that each permutation expressed by~$T$ is also expressed by $T_s^{\exit{(i-1)}}$.
  Moreover, invariant~(I2) holds locally for~$T$.
  
  The marked leaf-children of any half-marked node~$x$
  of~$T_s^{\exit{(i-1)}}$ can be placed anywhere before, between, or
  after its marked children, but not before~$y$ (since~$y$ has both
  marked and unmarked children).  Symmetrically, the unmarked
  leaf-children of any half-marked node~$x$ of~$T_s^{\exit{(i-1)}}$
  can be placed anywhere before, between, or after its unmarked
  children, but not after~$y$.
  Hence, each permutation expressed by $T_s^{\exit{(i-1)}}$ that has the neighbors of~$a_i$ as a suffix of the subsequence of its active vertices
  is also expressed by~$T$.
  Moreover, invariant~(I1) holds locally for~$T$.
  
  Thus, the permutations expressed by~$T$ are exactly those expressed by $T_s^{\exit{(i-1)}}$ that have the neighbors of~$a_i$ as a suffix of their active subsequence.

	Now, we create the data structure $\mathcal{T}^{\enter{i}}$; see Fig.~\ref{fig:stickA-construction-enter}.
	We set $T_{1}^{\enter{i}}=T_{1}^{\exit{(i-1)}},\ldots,$ $T_{s-1}^\enter{i}=T_{s-1}^{\exit{(i-1)}}$.
  Clearly, both invariants hold locally for $T_{1}^{\enter{i}},\ldots,T_{s-1}^{\enter{i}}$.
  Next, we create a new semi-ordered tree $T_s^{\enter{i}}$ as follows.
  The tree~$T_s^{\enter{i}}$ gets a new root~$r$, to which we
  attach~$T$ and a new vertex~$z$, in this order.  Then we hang
  $T_{s+1}^{\exit{(i-1)}}, \dots, T_{k_{i-1}}^{\exit{(i-1)}}$ from~$z$
  in this order.  We further make the entering vertices
  leaf-children of~$z$.  Note that this allows them to mix freely
  before, after, or between the components $G_{s+1}^{\exit{(i-1)}},
  \ldots, G_{k_{i-1}}^{\exit{(i-1)}}$.  Finally, we set 
  $\mathcal{T}^{\enter{i}}=(T_{1}^{\enter{i}},\ldots,T_{s}^{\enter{i}})$.

  In this way, the order of the components $G_{1}^{\exit{(i-1)}},\ldots,G_{k_{i-1}}^{\exit{(i-1)}}$ of $G^{\exit{(i-1)}}$ is maintained in the data structures for~$G^{\enter{i}}$. 
  In~$T_{s}^{\enter{i}}$, both invariants clearly hold for the non-leaf children of~$z$ and, as argued above, also for~$T$.
  Furthermore, we ensure that the entering vertices can be placed exactly before, after, or between the components of $G^{\exit{(i-1)}}$ that are completely adjacent to~$a_i$. 
	Hence, both invariants hold for~$\mathcal{T}^{\enter{i}}$. 
	
	The decision problem of \stickA can easily be solved by this algorithm.
  Namely, by our invariants, any permutation~$\sigma_B$ expressed by $\mathcal{T}^{\exit{|A|}}$ occurs as a permutation of the horizontal sticks in a \stickA representation of~$G$.
	Thus, executing our algorithm for \stickAB on $\sigma_A$ and~$\sigma_B$ gives us a stick representation of~$G$.
	This concludes the proof of correctness for the connected case.
	
	\myparagraph{Disconnected graphs.}
	To handle disconnected graphs, we first identify the connected components $H_1, \ldots, H_t$ of $G$. 
	We label each element of $A$ by the index of the component to which it belongs. 
	Now, observe that if $\sigma_A$ contains a pattern of
        indices~$a$ and~$a'$ that alternate as in $aa'aa'$, then the
        given \stickA instance does not admit a solution.  Otherwise,
        we can treat each component separately by our algorithm, and
        then nest the resulting representations. 
    Namely, for each connected component $H_r$, we run our \stickA algorithm (on $\sigma_A$ restricted to $H_r$) to obtain a realizable permutation $\sigma_{B_r}$ of the horizontal sticks of $H_r$. 
    Now, since our connected components avoid the pattern $aa'aa'$, there is natural hierarchy of these components which we can use to obtain a total order $\sigma_B$ on the horizontal sticks of $G$ from the permutations $\sigma_{B_1}, \ldots, \sigma_{B_t}$. 
    Finally, we can use this $\sigma_B$, the given $\sigma_A$, and $G$ as an input to our \stickAB algorithm to obtain a representation.  
	
	\myparagraph{Running time.}
	The size of each data structure $\mathcal{T}^p$ is in $O(|B^p|)$ since 
	there are no degree-2 vertices in the trees and each leaf corresponds to a vertex in~$B^p$. 
	In each event, the transformations can clearly be done in time proportional to the size of the data structures. 
	Since $|B^p|\le |B|$ for each~$p$ and there are $2|A|$ events, 
        we get the following running time.

\begin{thm}\label{thm:sticka}
	\stickA can be solved in $O(|A|\cdot|B|)$ time.
\end{thm}

\section{Sticks of Fixed Lengths}
\label{sec:fixed}

In this section, we consider the case that, for each vertex of the
input graph, its stick length is part of the input and fixed.  We
denote the variants of this problem by \stickfix if the order of the
sticks is not restricted, by \stickfixA if~$\sigma_A$ is given, and by
\stickfixAB if~$\sigma_A$ and~$\sigma_B$ are given.
Unlike in the case with variable stick
lengths, all three new variants are NP-hard; see
Sections~\ref{sec:fixed-length-np-complete}
and~\ref{sec:fixed-length-a-given}.  Surprisingly, \stickfixAB can be
solved efficiently by a simple linear program
if the input graph contains no isolated vertices
(i.e., vertices of degree~0); see Section~\ref{sec:fixed-length-a-b-given-no-isolated}.
With our linear program, we can check the feasibility of any instance
of~\stickfix if we are given a total order of the sticks on the ground
line.  With our NP-hardness results, this implies
NP-completeness of \stickfix, \stickfixA, and \stickfixAB.

\subsection{\stickfix}
\label{sec:fixed-length-np-complete}

We show that \stickfix is NP-hard by reduction from
\threePartition, which is strongly NP-hard~\cite{Garey1979}.  In
\threePartition, one is given a multiset~$S$ of $3m$ integers
$s_1,\dots,s_{3m}$ such that, for $i \in \{1,\dots,3m\}$,
$C/4<s_i<C/2$, where $C=(\sum_{i=1}^{3m} s_i)/m$, and the task is to
decide whether~$S$ can be split into $m$ sets of three integers, each
summing up to~$C$.

\begin{figure}[t]
	\centering
	\begin{subfigure}[t]{0.48 \linewidth}
		\centering
		\includegraphics{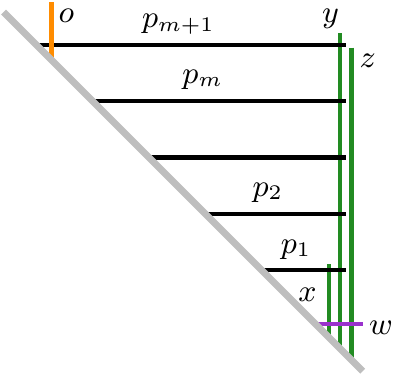}
		\caption{frame providing the pockets}
		\label{fig:cage}
	\end{subfigure}
	\hfill
	\begin{subfigure}[t]{0.48 \linewidth}
		\centering
		\includegraphics{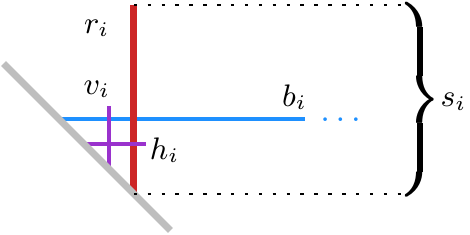}
		\caption{number gadget for the number $s_i$}
		\label{fig:gadget}
	\end{subfigure}
	\caption{Gadgets of our reduction from \threePartition to \stickfix.}
	\label{fig:frame}
\end{figure}

\begin{thm}
	\label{thm:stick-fixedlengths}
	\stickfix is NP-complete.
\end{thm}

\begin{proof}
	As mentioned at the beginning of this section, NP-membership follows from our linear program (see Theorem~\ref{thm:stick-fixed-lp} in Section~\ref{sec:fixed-length-a-b-given-no-isolated})
	to test the feasibility of any instance of~\stickfix when given a total order of the sticks on the ground line.
	
	To show NP-hardness, we describe a polynomial-time reduction from \threePartition to \stickfix.
	Given a \threePartition instance~$I = (S, C, m)$, we construct a fixed cage-like frame structure and introduce a number gadget for each number of~$S$.
	%
	A sketch of the frame is given in Fig.~\ref{fig:cage}.
	The purpose of the frame is to provide pockets, which will host our number gadgets (defined below).
	We add two long vertical (green) sticks~$y$ and $z$ of length~$mC+1+2\varepsilon$ and a shorter vertical (green) stick~$x$ of length~$1$ that are all kept together by a short horizontal (violet) stick~$w$ of some length~$\varepsilon \ll 1$.
	We use $m+1$ horizontal (black) sticks $p_1, p_2, \dots, p_{m+1}$ to separate the pockets.
	Each of them intersects~$y$ but not~$z$ and has a specific length such that the distance between two of these sticks is~$C \pm \varepsilon$.
		
	Additionally, $p_1$ intersects~$x$ and $p_{m+1}$ intersects a vertical (orange) stick~$o$ of length~$2C$.
	We use $x$ and $o$ to prevent the number gadgets from being placed below the bottommost and above the topmost pocket, respectively.
	It does not matter on which side of~$y$ the stick~$x$
	ends up since each $b_i$ of a number gadget
	intersects~$y$ but neither~$x$ nor~$z$.

	\label{ssec:gadget}
	For each number~$s_i$ in~$S$, we construct a number gadget; see Fig.~\ref{fig:gadget}.
	We introduce a vertical (red) stick~$r_i$ of length~$s_i$.
	Intersecting~$r_i$, we add a horizontal (blue) stick~$b_i$ of length at least $mC+2$.
	The stick~$b_i$ intersects $y$ and~$z$, but neither $x$ nor~$o$.
	Due to these adjacencies, every number gadget can only be placed in one of the $m$ pockets defined by $p_1, \dots, p_{m+1}$.
	It cannot span multiple pockets.
	We require that $r_i$ and~$b_i$ intersect each other close to their foot points,
	so we introduce two short (violet) sticks~$h_i$ and $v_i$---one horizontal, the other vertical---of lengths~$\varepsilon$;
	they intersect each other, $h_i$ intersects~$r_i$, and $v_i$ intersects~$b_i$.
	
	Given a yes-instance $I=(S,C,m)$ and a valid 3-partition~$P$ of $S$, the graph obtained by our reduction is realizable.
	Construct the frame as described before and place the number gadgets into the pockets according to~$P$.
	Since the lengths of the three number gadgets' $r_i$ sum up to~$C \pm 3 \varepsilon$, all three can be placed into one pocket. 
	After distributing all number gadgets, we have a stick representation.
	
	Given a stick representation of a graph $G$ obtained from our reduction, we can obtain a valid solution of the corresponding \threePartition instance $I=(S,C,m)$ as follows.
	Clearly, the shape of the frame is fixed, creating
	$m$~pockets.
	Since the sticks~$b_1,\dots,b_{3m}$ are incident to $y$ and
	$z$ but neither to~$x$ nor to~$o$, they can end up 
	inside any of the pockets. 
	In the y-dimension, each two number gadgets of numbers $s_k$
	and $s_\ell$ overlap at most on a section of
	length~$\varepsilon$; otherwise $r_k$ and~$b_\ell$ or
	$r_\ell$~and~$b_k$ would intersect.
	Each pocket hosts precisely three number gadgets: we have $3m$
	number gadgets, $m$ pockets, and no pocket can contain four
	(or more) number gadgets; otherwise, there would be a number
	gadget of height at most $(C + \varepsilon)/4 + 2\varepsilon$,
	contradicting the fact that $s_i$ is an integer with $s_i>C/4$.
	In each pocket, the height of the number gadgets would be too large if the three corresponding numbers of~$S$ would sum up to~$C + 1$ or more.
	Thus, the assignment of number gadgets to pockets defines a valid 3-partition of~$S$.
\end{proof}

The sticks of lengths~$s_1, \dots, s_{3m}$ can be simulated by paths of sticks with length~$\varepsilon$ each.
We refer to them as \emph{$\varepsilon$-paths}.
Employing them in our reduction, it suffices to use only three distinct stick lengths.
Like a spring, an $\varepsilon$-path can be stretched (Fig.~\ref{fig:epsilonPathStretched}) and compressed (Fig.~\ref{fig:epsilonPathCompressed}) up to a specific length.
We will exploit the following properties regarding the minimum and the maximum size of an $\varepsilon$-path.

\begin{figure}[t]
	\centering
	\begin{subfigure}[t]{0.32\linewidth}
		\centering
		\includegraphics[page=1]{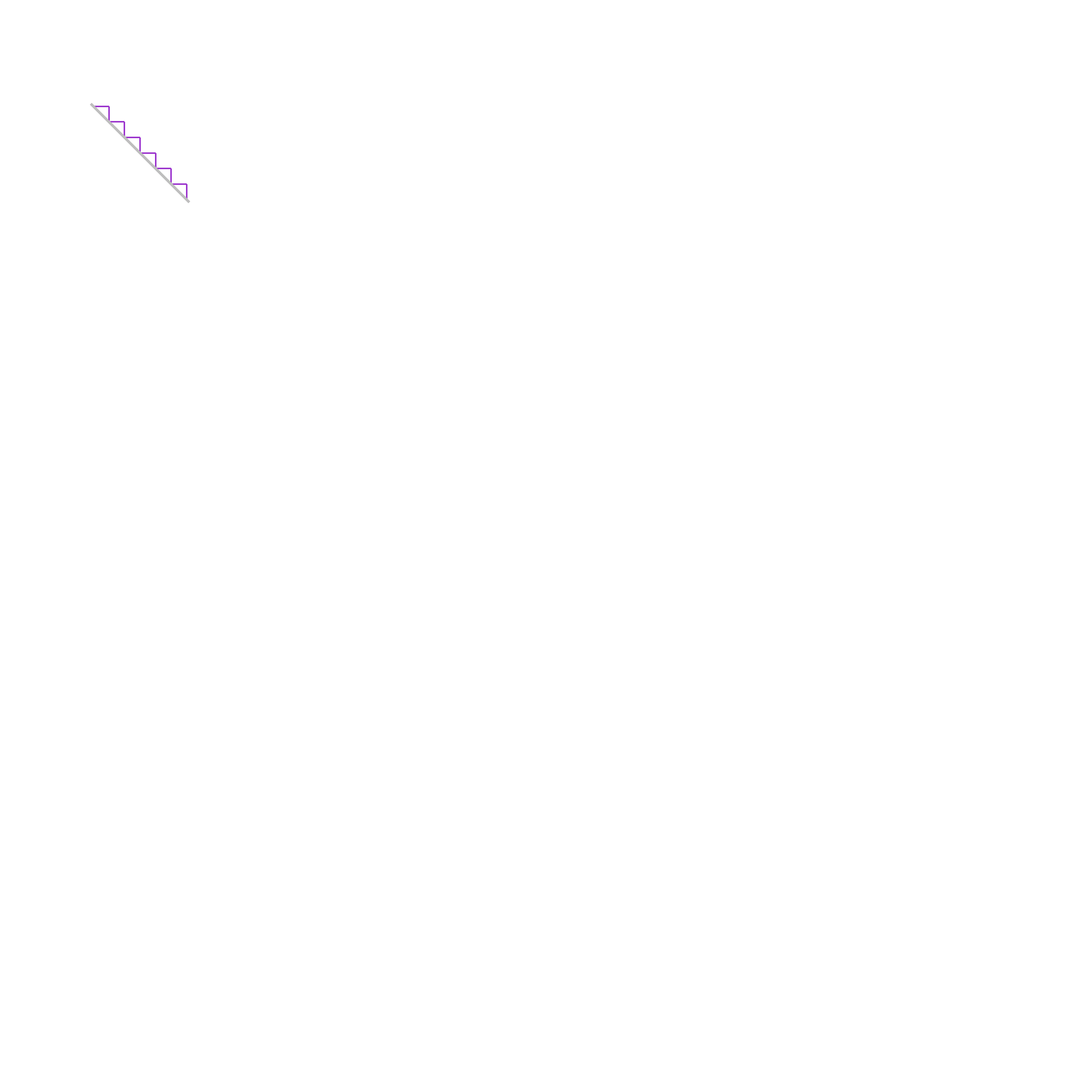}
		\caption{stretched}
		\label{fig:epsilonPathStretched}
	\end{subfigure}
	\hfill
	\begin{subfigure}[t]{0.32\linewidth}
		\centering
		\includegraphics[page=2]{epsilon-path}
		\caption{medium}
		\label{fig:epsilonPathInBetween}
	\end{subfigure}
	\hfill
	\begin{subfigure}[t]{0.32\linewidth}
		\centering
		\includegraphics[page=3]{epsilon-path}
		\caption{compressed}
		\label{fig:epsilonPathCompressed}
	\end{subfigure}
	\caption{Three stick representations of an $\varepsilon$-path with twelve sticks.}
	\label{fig:epsilonPath}
\end{figure}

\begin{figure}[t]
	\centering
	\includegraphics[page=4, scale=0.94]{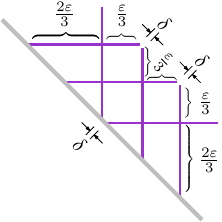}
	\caption{Construction of a compressed stick representation of an $\varepsilon$-path}
	\label{fig:epsilonPath6Sticks}
\end{figure}

\begin{lemma}
	\label{lem:epsilonPaths}
	There is a stick representation of a $2n$-vertex $\varepsilon$-path with height and width~$n \varepsilon$ and another stick representation with height and width~$\frac{n+2}{3} \varepsilon + \delta$ for any $\delta > 0$ and $n \geq 3$.
	Any stick representation of a $2n$-vertex $\varepsilon$-path has height and width in the range~$\left(\frac{n}{3}, n \right] \varepsilon$.
\end{lemma}

\begin{proof}
	We can arrange our sticks such that the foot points or the end points of two adjacent sticks touch each other (see Fig.~\ref{fig:epsilonPathStretched}).
	This construction has height and width $n \varepsilon$ and, clearly, this is the maximum width and height for a $2n$-vertex $\varepsilon$-path.
	
	For the compressed $\varepsilon$-paths, we first describe a construction that has the specified width and height and, second, we show the lower bound.
	
	The following construction is depicted in Fig.~\ref{fig:epsilonPath6Sticks} for $n = 3$.
	Set the foot point of the first vertical stick in the path to $y=0$ and the foot point of the third stick, which is also vertical, to $y = \varepsilon/3$.
	For each $i \in \{2, \dots, n-1\}$, set the foot point of the $(2 i - 2)$-th stick (horizontal) to $y = i \varepsilon/3 + (i-2) \delta / (n-2)$ and the foot point of the $(2 i + 1)$-th stick (vertical) to $y = i \varepsilon/3 + (i-1) \delta / (n-2)$.
	Set the foot point of the $(n-2)$-th stick to $y = n \varepsilon / 3 + \delta$, and the foot point of the last stick to $y = (n+1) \varepsilon / 3 + \delta$.
	Observe that this construction has width and height~$\frac{n+2}{3} \varepsilon + \delta$ and is a valid stick representation of a $2n$-vertex $\varepsilon$-path.
	
	Consider the $i$-th stick of an $\varepsilon$-path.
	On the one side of the line through this stick, there is the $(i-3)$-th stick, and on the other, there is the $(i+3)$-th stick.
	E.g., the second stick is to the right of the fifth stick and the eighth stick is to the left of the fifth stick.
	Since all sticks have length $\varepsilon$ and non-adjacent sticks are not allowed to touch each other, the 1st, 4th, 7th, \dots, $(6k-2)$-th stick for $k \in \mathbb{N}$ form a zigzag chain of width and height strictly greater than $k \varepsilon$.
	The same holds for the 2nd, 5th, \dots stick and the 3rd, 6th, \dots stick.
	Thus, for an $\varepsilon$-path of~$2n$ sticks, we have width and height strictly greater than $\left\lceil\frac{2n}{6}\right\rceil \varepsilon \geq \frac{n}{3} \varepsilon$.
\end{proof}

\begin{figure}[t]
	\centering
	\begin{subfigure}[t]{0.53\linewidth}
		\centering
		\includegraphics[page=1]{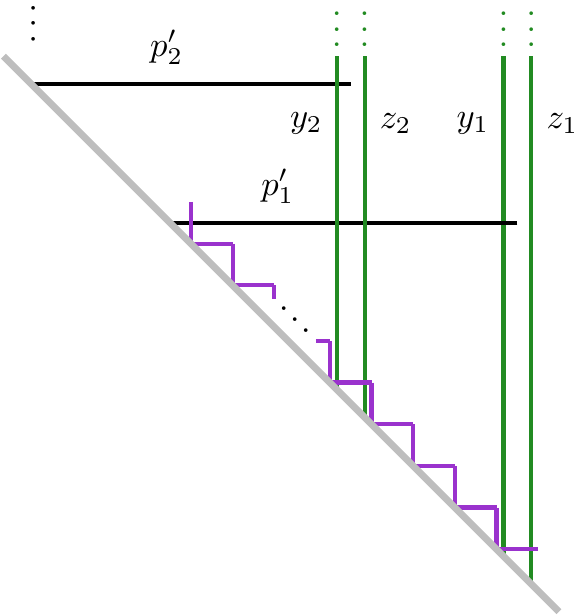}
		\caption{frame providing the pockets}
		\label{fig:3cage}
	\end{subfigure}
	\hfill
	\begin{subfigure}[t]{0.42\linewidth}
		\centering
		\includegraphics[page=1]{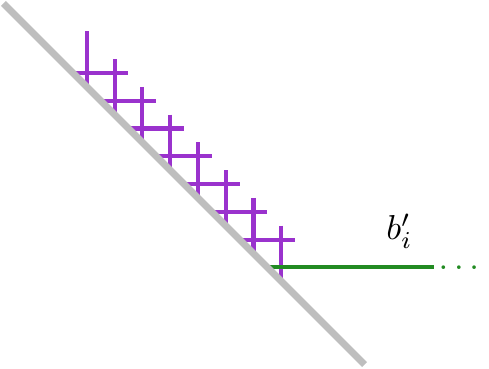}
		\caption{number gadget for the number~$s_i$}
		\label{fig:3gadget}
	\end{subfigure}
	\caption{Gadgets of our reduction from \threePartition to \stickfix with three stick lengths.}
	\label{fig:3lengths}
\end{figure}

\begin{cor}
	\label{cor:stick-fixed-3-lenghts}
	\stickfix with only three different stick lengths is NP-complete.
\end{cor}

\begin{proof}
	We modify the reduction from \threePartition to \stickfix described in the proof of Theorem~\ref{thm:stick-fixedlengths} such that we use only three distinct stick lengths.
	We use the three lengths $\varepsilon$, $Cm$, and $3Cm$ (or longer, e.g. $\infty$).
	In Fig.~\ref{fig:3lengths}, sticks of these lengths are violet, black, and green, respectively.
	
	First, we describe the modifications of the frame structure, which are also depicted in Fig.~\ref{fig:3cage}.
	Instead of the vertical (green) sticks $x$, $y$, and $z$ used to fix all pockets, we have two vertical sticks~$y_j$ and $z_j$ of length~$3Cm$ for $j \in \{1, \dots, m+1\}$.
	Instead of the sticks $p_1, \dots, p_{m+1}$ of different lengths, we use horizontal (black) sticks~$p'_1, \dots, p'_{m+1}$ each with length~$Cm$ to separate the pockets.
	The stick~$p'_j$ intersects $y_k, z_k$ for all $k \in \{j+1, \dots, m+1\}$ and $y_j$ but not $z_j$.
	All pairs $(y_j, z_j)$ are kept together by a stick of length~$\varepsilon$.
	For each two neighboring pairs $(y_j, z_j)$ and $(y_{j+1}, z_{j+1})$, these sticks of length~$\varepsilon$ are connected by an $\varepsilon$-path of $2 C / \varepsilon$ sticks.
	According to Lemma~\ref{lem:epsilonPaths}, this effects a maximum distance of $(C / \varepsilon) \cdot \varepsilon \pm \varepsilon = C \pm \varepsilon$ between each two pairs of $(y_j, z_j)$ and $(y_{j+1}, z_{j+1})$.
	Accordingly, the pockets separated by the sticks~$p'_1, \dots, p'_{m+1}$ have height at most~$C \pm 2 \varepsilon$, similar as in the proof of Theorem~\ref{thm:stick-fixedlengths}.
	We keep the vertical (orange) stick~$o$ as in Fig.~\ref{fig:cage} to prevent number gadgets from being placed above the topmost pocket, but now $o$ has length~$3Cm$.
	
	Second, we describe the modifications of the number gadgets for each number~$s_i$ for $i \in \{1, \dots, 3m\}$, which are also depicted in Fig.~\ref{fig:3gadget}.
	We keep a long stick $b'_i$ similar to $b_i$---now with length $3Cm$.
	We replace $r_i$ (together with $h_i$ and $v_i$) by an $\varepsilon$-path of $6 s_i / \varepsilon - 4$ sticks.
	We make the first stick of the $\varepsilon$-path intersect $b'_i$.
	By Lemma~\ref{lem:epsilonPaths}, this $\varepsilon$-path has a stick representation with height $s_i + \delta$ for any $\delta > 0$, but there is no stick representation with height $s_i - \frac{2}{3} \varepsilon$ or smaller.
	Clearly, these number gadgets can only be placed into one pocket since none of their sticks intersects a $p'_j$ for $j \in \{1, \dots, m+1\}$.
	
	Hence, we can represent a yes-instance of \threePartition as such a stick graph if and only if the $\varepsilon$-paths of the number gadgets are (almost) as much compressed as possible (to make the number gadgets small enough) and the $\varepsilon$-paths between the $(y_j, z_j)$-sticks are (almost) as much stretched as possible (to make the pockets tall enough).
	Using this, the proof is the same as in Theorem~\ref{thm:stick-fixedlengths}.
\end{proof}

\subsection{\stickfixA and \stickfixAB}
\label{sec:fixed-length-a-given}

We show that \stickfixA and \stickfixAB are NP-hard by reducing from
\monThreeSat, which is NP-hard~\cite{l-tscr-DAM97}.  
In \monThreeSat, one is given a Boolean formula~$\Phi$ in CNF where each clause contains three distinct literals that are all positive or all negative.
The task is to decide whether~$\Phi$ is satisfiable.

\begin{thm}
	\label{thm:sticka-fixedlengths}
	\stickfixA is NP-complete.
\end{thm}

\begin{figure}[t]
	\centering
	
	\begin{subfigure}[t]{.45 \linewidth}
		\centering
		\includegraphics[page=1, width=\textwidth]{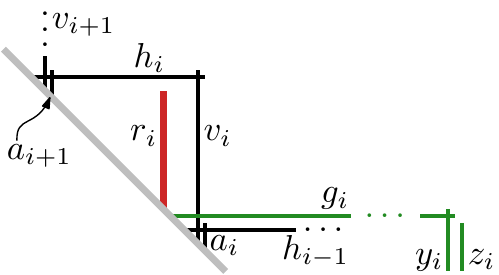}
		\caption{variable gadget set to false}
		\label{fig:variableMon3SAT-false}
	\end{subfigure}
	\hfill
	\begin{subfigure}[t]{.45 \linewidth}
		\centering
		\includegraphics[page=2, width=\textwidth]{variable-gadget-Mon3SAT}
		\caption{variable gadget set to true}
		\label{fig:variableMon3SAT-true}
	\end{subfigure}
	
	\caption{Variable gadget in our reduction from \monThreeSat to \stickfixA.}
	\label{fig:variableMon3SAT}
\end{figure}

\noindent \textbf{Proof:}
Recall that, as mentioned before, NP-membership follows from our linear program (see Theorem~\ref{thm:stick-fixed-lp} in Section~\ref{sec:fixed-length-a-b-given-no-isolated})
to test the feasibility of any instance of~\stickfix when given a total order of the sticks on the ground line.
To show NP-hardness, we describe a polynomial-time reduction from \monThreeSat to \stickfixA.
A schematization of our reduction is depicted in
Figs.~\ref{fig:variableMon3SAT}, \ref{fig:reductionFromMon3SAT}, and
\ref{fig:clauseGadgetsMon3SAT}.

Given a \monThreeSat instance $\Phi$ over variables $x_1, \dots, x_n$,
we construct a \emph{variable gadget} for each variable
 as depicted in Fig.~\ref{fig:variableMon3SAT}.
Inside a (black) \emph{cage}, there is a vertical (red) stick~$r_i$ with length~1 and from the inside a long horizontal (green) stick~$g_i$ leaves this cage.
We can enforce the structure to be as in Fig.~\ref{fig:variableMon3SAT} as follows.
We prescribe the order~$\sigma_A$ of the vertical sticks as in Fig.~\ref{fig:variableMon3SAT}.
Since~$a_{i+1}$ has length $\varepsilon \ll 1$, 
the horizontal (black) stick~$h_i$ intersects the two vertical (black) sticks~$v_{i+1}$ and~$a_{i+1}$ close to its foot point.
Since we have prescribed $\sigma_A(a_{i+1}) < \sigma_A(r_i) < \sigma_A(v_i)$, we see that~$r_i$ is inside the cage bounded by~$h_i$ and~$v_i$ and fixes its height---as it does not intersect $h_i$---making the sticks~$h_i$ and~$v_i$ intersect close to their end points (both have length $1 + 2 \varepsilon$).
Moreover, $r_i$ cannot be below $h_{i-1}$ because~$a_i$ is shorter than~$r_i$ and intersects~$h_{i-1}$ to the right of $r_i$.
This leaves the freedom of placing~$g_i$ above or below $r_i$ (as $g_i$ does not intersect $r_i$) but still with its foot point
inside the cage formed by $h_i$ and~$v_i$ because it intersects $v_i$, but neither $v_{i-1}$ nor $v_{i+1}$.

We say that the variable~$x_i$ is set to false if the foot point of~$g_i$ is below the foot point of~$r_i$, and true otherwise.
For each~$x_i$, we add two long vertical (green) sticks~$y_i$ and~$z_i$ that we keep close together by using a short horizontal (violet) stick of length~$\varepsilon$ (see Fig.~\ref{fig:reductionFromMon3SAT} on the bottom right).
We make~$g_i$ intersect~$y_i$ but not $z_i$.
The three sticks $g_i$, $y_i$, and $z_i$ get the same length~$\ell_i$.
Hence, $y_i$ and $g_i$ intersect each other close to their end points as otherwise~$g_i$ would intersect~$z_i$.
We choose $\ell_1$ sufficiently large such that the foot point of $y_1$ is to the right of the clause gadgets (see Fig.~\ref{fig:reductionFromMon3SAT}) and for each $\ell_i$ with $i \geq 2$, we set $\ell_i = \ell_{i-1} + 1 + 3 \varepsilon$.

Now compare the end points of $g_i$ when $x_i$ is set to false and when $x_i$ is set to true relative to the (black) cage structure.
When $x_i$ is set to true, the end point of $g_i$ is $1 \pm 2 \varepsilon$ above and $1 \pm 2 \varepsilon$ to the left compared to the case when $x_i$ is set to false.
Observe that, since $g_i$ and $y_i$ intersect each other close to their end points, this offset is also pushed to $y_i$ and $z_i$ and their foot points.
Consequently, the position of the foot point of~$y_i$ (and~$z_i$) differs by $1 \pm 2 \varepsilon$ relative to the (black) frame structure
depending on whether $x_i$ is set to true or false.
Our choice of $\ell_i$ allows this movement.
In other words, no matter which truth value we assign to each $x_i$, there is a stick representation of the variable gadgets respecting~$\sigma_A$.

\begin{wrapfigure}[0]{R}{0.89 \textwidth}
	\centering
	\vspace{-3.5ex}
	\caption{Illustration of our reduction from \monThreeSat to \stickfixA}
	\label{fig:reductionFromMon3SAT}

	\vspace{-1.5ex}
	\includegraphics[page=9, width=0.9 \textwidth]{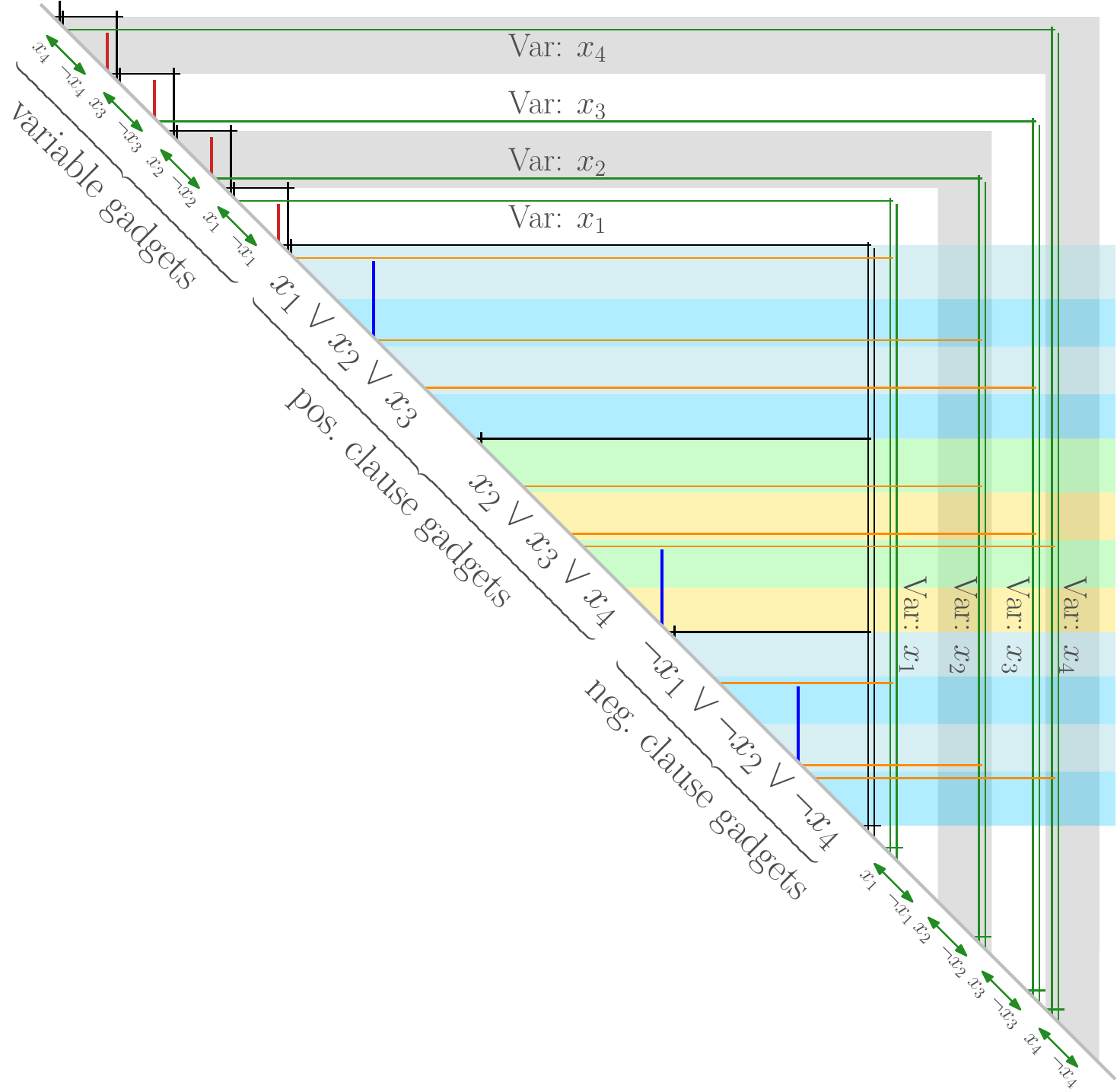}
\end{wrapfigure}	
\null\hfill\smash{
	\raisebox{\dimexpr-\height+\baselineskip}{%
}}%
\vspace{19ex}
\par\vspace*{\dimexpr-\baselineskip-\parskip}
\parshape 21 0pt 0.01\textwidth
0pt \dimexpr0.01\textwidth+5\baselineskip\relax
0pt \dimexpr0.01\textwidth+6\baselineskip\relax
0pt \dimexpr0.01\textwidth+7\baselineskip\relax
0pt \dimexpr0.01\textwidth+8\baselineskip\relax
0pt \dimexpr0.01\textwidth+9\baselineskip\relax
0pt \dimexpr0.01\textwidth+10\baselineskip\relax
0pt \dimexpr0.01\textwidth+11\baselineskip\relax
0pt \dimexpr0.01\textwidth+12\baselineskip\relax
0pt \dimexpr0.01\textwidth+13\baselineskip\relax
0pt \dimexpr0.01\textwidth+14\baselineskip\relax
0pt \dimexpr0.01\textwidth+15\baselineskip\relax
0pt \dimexpr0.01\textwidth+16\baselineskip\relax
0pt \dimexpr0.01\textwidth+17\baselineskip\relax
0pt \dimexpr0.01\textwidth+18\baselineskip\relax
0pt \dimexpr0.01\textwidth+19\baselineskip\relax
0pt \dimexpr0.01\textwidth+22\baselineskip\relax
0pt \dimexpr0.01\textwidth+23\baselineskip\relax
0pt \dimexpr0.01\textwidth+24\baselineskip\relax
0pt \dimexpr0.01\textwidth+25\baselineskip\relax
0pt \textwidth
\noindent
For~each clause, we add a \emph{clause gadget} (see Fig.~\ref{fig:clauseGadgetsMon3SAT}) as shown in Fig.~\ref{fig:reductionFromMon3SAT}.
It is a stripe that is bounded by horizontal (black) sticks on its top and bottom.
To fix the height of each stripe, we introduce two vertical (black) sticks that we keep close together by a short horizontal (black) stick of length~$\varepsilon$.
We make each horizontal (black) stick intersect only the first of these vertical (black) sticks to obtain clause gadgets of height of $4 + 2 \varepsilon \pm \varepsilon$.
Moreover, we make the topmost horizontal (black) stick intersect $a_1$ and~$v_1$ to keep them connected to the variable gadgets.
We (virtually) divide each clause gadget into four horizontal sub-stripes of height $\geq 1$.
For \emph{positive clause gadgets} corresponding to all-positive clauses, we leave the bottommost sub-stripe empty;
for \emph{negative clause gadgets} corresponding to all-negative clauses, we leave the topmost sub-stripe empty.
We add three horizontal (orange) sticks---one per remaining horizontal sub-stripe---and assign them bijectively to the variables of the clause.
We make each horizontal (orange) stick~$o$ that is assigned to $x_i$ intersect $y_i$ and all $y_j$ and $z_j$ for $j < i$, but not $z_i$ or $y_k$ or $z_k$ for any $k > i$.
Thus, $o$ intersects $y_i$ close to $o$'s end point.
We choose the length of each such~$o$ so that its foot point is at the bottom of its sub-stripe if $x_i$ is set to false or is at the top of its sub-stripe if $x_i$ is set to true.
Within the positive and the negative clause gadgets, this gives us two times eight possible configurations of the orange sticks depending on the truth assignment of the three variables of the clause (see Fig.~\ref{fig:clauseGadgetsMon3SAT}).
\begin{figure}[t]
	\centering
	\includegraphics[page=1, width=.96\textwidth]{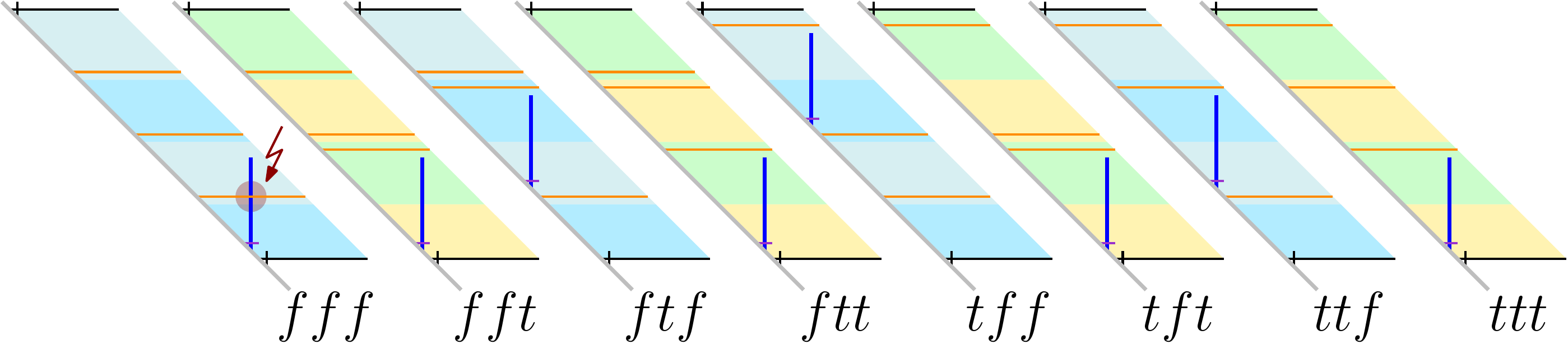}
	
	\caption*{Positive clause gadget (empty sub-stripe at the bottom)}
	
	\bigskip
	
	\centering
	\includegraphics[page=1, width=.96\textwidth]{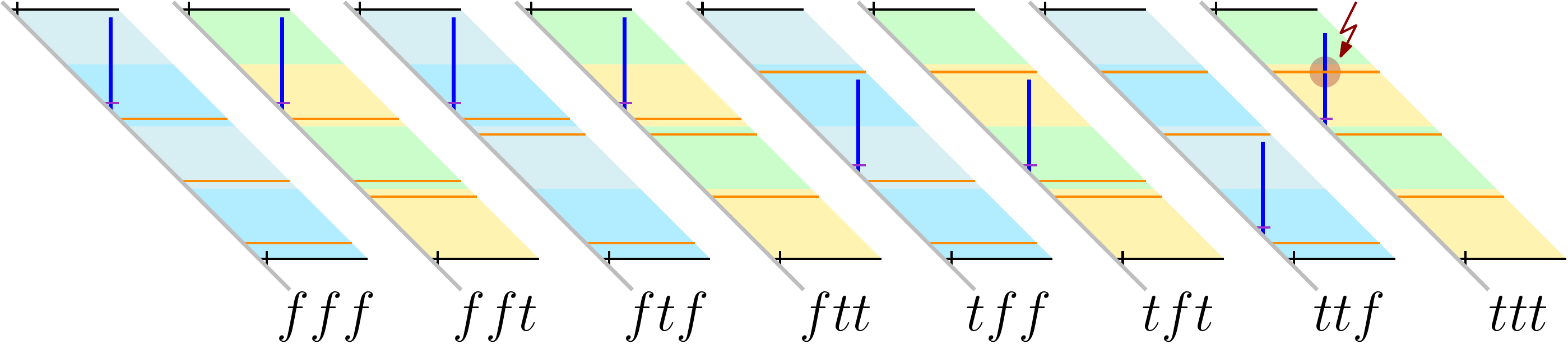}
	
	\caption*{Negative clause gadget (empty sub-stripe at the top)}
	
	\caption{%
		Clause gadget in our reduction from \monThreeSat to \stickfixA.
		Here, a clause gadget for each of the eight possible truth assignments of a \monThreeSat clause is depicted.
		In particular, it shows that the (blue) isolated vertical stick fits inside the gadget if and only if the corresponding clause is satisfied.
		E.g., $tft$ means that the first variable is set to true, the second to false, and the third to true.}
	\label{fig:clauseGadgetsMon3SAT}
\end{figure}
Within each clause gadget, we have a vertical (blue) stick~$b$ of length~2, which does not intersect any other stick.
Each horizontal (black) stick that bounds a clause gadget intersects a short vertical (black) stick of length~$\varepsilon$ to force~$b$ into its designated clause gadget.

Clearly, if $\Phi$ is satisfiable, there is a stick representation of the \stickfixA instance obtained from~$\Phi$ by our reduction by placing the sticks as described before (see also Fig.~\ref{fig:reductionFromMon3SAT}).
In particular, each blue stick can be placed in one of the ways depicted in Fig.~\ref{fig:clauseGadgetsMon3SAT}.

On the other hand, if there is a stick representation of the \stickfixA instance obtained by our reduction, $\Phi$ is satisfiable.
As argued before, the shape of the (black) frame structure of all gadgets is fixed by the choice of the adjacencies and lengths in the graph and~$\sigma_A$.
The only flexibility is, for each $i \in \{1, \dots n\}$, whether~$g_i$ has its foot point above or below~$r_i$.
This enforces one of eight distinct configurations per clause gadget.
As depicted in Fig.~\ref{fig:clauseGadgetsMon3SAT}, precisely the configurations that correspond to satisfying truth assignments are realizable.
Thus, we can read a satisfying truth assignment of~$\Phi$ from the
variable gadgets.

Our reduction can obviously be implemented in polynomial time. 
\qed
\medskip

In our reduction, we enforce an order of the horizontal sticks.
So, prescribing~$\sigma_B$ makes it even easier to enforce this structure.
Hence, we can use exactly the same reduction for \stickfixAB.

\begin{cor}
	\label{cor:stickab-fixedlengths-isolated}
	\stickfixAB (with isolated vertices in~$A$ or~$B$) is NP-complete.
\end{cor}

\begin{proof}
	Given a \monThreeSat instance~$\Phi$,
	consider the construction described in the proof of Theorem~\ref{thm:sticka-fixedlengths}.
	We use the same graph, the same stick lengths and the same ordering~$\sigma_A$ of the vertical sticks.
	Now, we additionally prescribe the order of the remaining horizontal sticks as depicted in Fig.~\ref{fig:reductionFromMon3SAT} via~$\sigma_B$.
	This defines an instance of~\stickfixAB.
	
	Clearly, the ordering of the horizontal sticks~$\sigma_B$ neither affects the placement of the vertical isolated (red) sticks inside a variable gadget nor does it affect the placement of the vertical isolated (blue) sticks inside a clause gadget.
	Moreover, there was only one possible ordering of the horizontal sticks in the construction described in the proof of Theorem~\ref{thm:sticka-fixedlengths}.
	Thus, its correctness proof applies here as well.
\end{proof}

The reduction we described before uses isolated vertices inside the variable and the clause gadgets.
In the case of \stickfixAB, this is indeed necessary to show NP-hardness.
This is not true for \stickfixA, which remains NP-hard (and hence is NP-complete due to our linear program) even without isolated sticks.

\begin{figure}[t]
	\centering
	
	\begin{subfigure}[t]{.45 \linewidth}
		\centering
		\includegraphics[page=3, width=\textwidth]{variable-gadget-Mon3SAT}
		\caption{variable gadget with isolated stick}
		\label{fig:variableMon3SAT-with-isolated}
	\end{subfigure}
	\hfill
	\begin{subfigure}[t]{.48 \linewidth}
		\centering
		\includegraphics[page=5, width=\textwidth]{variable-gadget-Mon3SAT}
		\caption{variable gadget without isolated stick}
		\label{fig:variableMon3SAT-without-isolated}
	\end{subfigure}
	
	\bigskip
	
	\begin{subfigure}[t]{.45 \linewidth}
		\centering
		\includegraphics[page=1, width=0.5\textwidth]{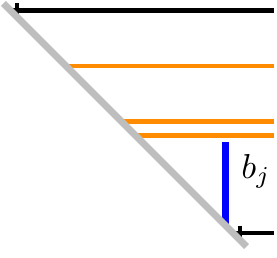}
		\caption{clause gadget with isolated stick}
		\label{fig:clauseMon3SAT-with-isolated}
	\end{subfigure}
	\hfill
	\begin{subfigure}[t]{.48 \linewidth}
		\centering
		\includegraphics[page=2, width=0.5\textwidth]{clause-gadget-without-isolated-stick}
		\caption{clause gadget without isolated stick}
		\label{fig:clauseMon3SAT-without-isolated}
	\end{subfigure}
	
	\caption{We add a short horizontal (violet) stick~$w_i$ (or $w'_j$) that intersects~$r_i$ (or $b_j$) to avoid isolated sticks in a variable (or clause) gadget for variable~$x_i$ (or clause~$C_j$).}
	\label{fig:gadgetsMon3SAT-with-and-without-isolated}
\end{figure}

\begin{cor}
	\label{cor:sticka-fixedlengths-without-isolated-v}
	\stickfixA without isolated vertices is NP-complete.
\end{cor}

\begin{proof}
	We use the same reduction as in the proof of Theorem~\ref{thm:sticka-fixedlengths},
	but we additionally add, for each isolated stick, a short
        stick of length~$\varepsilon \ll 1$ that only intersects the
        isolated stick; see
        Fig.~\ref{fig:gadgetsMon3SAT-with-and-without-isolated}.
	In each variable gadget, for the isolated vertical (red) stick~$r_i$, we add a short horizontal (violet) stick~$w_i$ of length~$\varepsilon$.
	Similarly, in each clause gadget, for the isolated vertical
        (blue) stick~$b_j$, we add a short horizontal (violet)
        stick~$w'_j$ of length~$\varepsilon$.
	After these additions, no isolated sticks remain.
	
	Observe that, for any placement of the isolated (red and blue) sticks inside their gadgets in the proof of Theorem~\ref{thm:sticka-fixedlengths}, we can add the new horizontal (violet) stick since it has length only~$\varepsilon$.
	Moreover, since these new (violet) sticks are horizontal, we do not get any new ordering constraints in the version \stickfixA.
	
	We therefore conclude that the rest of the proof still holds.
\end{proof}

\subsection{\stickfixAB without isolated vertices}
\label{sec:fixed-length-a-b-given-no-isolated}
In this section, we constructively show that \stickfix is efficiently solvable
if we are given a total order of the vertices in~$A \cup B$ on the ground line.
Note that if there is a stick representation for an instance of \stickAB (and consequently also \stickfixAB), the combinatorial order of the sticks on the ground line is always the same except for isolated vertices,
which we formalize in the following lemma.
The proof also follows implicitly from the proof of Theorem~\ref{thm:stickab}.

\begin{figure}[t]
	\centering
	\begin{subfigure}[t]{0.5 \linewidth}
		\centering
		\includegraphics[page=1]{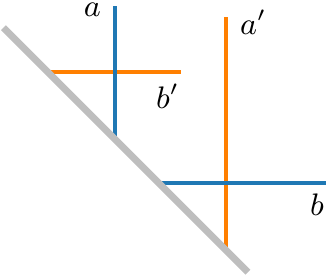}
		\caption{$a$ comes before $b$ (valid representation)}
		\label{fig:a-after-b}
	\end{subfigure}
	\hfill
	\begin{subfigure}[t]{0.45 \linewidth}
		\centering
		\includegraphics[page=2]{unique-ordering-stickab}
		\caption{$a$ comes after $b$ (no representation)}
		\label{fig:a-before-b}
	\end{subfigure}
	
	\caption{Trying to represent a subgraph of the edges $ab'$ and $a'b$ while respecting $\sigma_A$ and~$\sigma_B$.}
	\label{fig:STICKAB-without-isolated-is-unique}
\end{figure}

\begin{lem}
	\label{lem:STICKAB-without-isolated-is-unique}
	In all stick representations of an instance of \stickAB, the order of the vertices $A \cup B$ on the ground line is the same after removing all isolated vertices.
	This order can be found in time $O(|E|)$.
\end{lem}

\begin{proof}
	Assume there are stick representations $\Gamma_1$ and $\Gamma_2$ of the same instance of \stickAB without isolated vertices that have different combinatorial arrangements on the ground line.
	Without loss of generality, there is a pair of sticks $(a,b)
        \in A \times B$ such that in $\Gamma_1$, $a$ comes before~$b$,
        while in $\Gamma_2$, $a$ comes after~$b$ (see
        Fig.~\ref{fig:STICKAB-without-isolated-is-unique}).
	Clearly, $a$ and $b$ cannot be adjacent.
	Since $a$ is not isolated, there is a $b'$ that is adjacent to $a$ and comes before $b$.
	Analogously, there is an $a'$ that is adjacent to $b$ and comes after~$a$.
	In $\Gamma_2$, stick~$b$, stick~$a'$, and the ground line define
        a triangular region~$T$ (see Fig.~\ref{fig:a-before-b}), which
        completely contains $a$ since $a$ occurs between $b$ and $a'$,
        but is adjacent to neither of them.
	However, $b'$ is outside of $T$ as it comes before $b$.
	This contradicts $b$ and $a'$ being adjacent.
	The unique order can be determined in~$O(|E|)$ time as described in Section~\ref{sec:var_length}.
\end{proof}

We are given an instance of~\stickfix and a total order~$v_1,\ldots,v_n$ of the vertices ($n=|A|+|B|$) with stick lengths~$\ell_1,\ldots,\ell_n$.
We create a system of difference constraints, that is, a linear program $Ax\le b$
where each constraint is a simple linear inequality of the form $x_j-x_i\le b_k$,
with $n$ variables and $m\le 3n-1$ constraints.
Such a system can be modeled as a weighted graph with a vertex per variable $x_i$ and
a directed edge $(x_i,x_j)$ with weight~$b_k$ per constraint. 
The system has a solution if and only if there is no directed cycle of negative weight.
In this case, a solution can be found in $O(nm)$ time using the Bellman--Ford algorithm.

In the following, we describe how to construct such a linear program for \stickfix.
For each stick~$v_i$, we create a variable~$x_i$ that corresponds to
the x-coordinate of~$v_i$'s foot point on the ground line, with $x_1=0$.
To ensure the unique order, we add~$n-1$ constraints
$x_{i+1}-x_i\le -\varepsilon$ for some suitably small $\varepsilon$ and $i=1,\ldots,n-1$.

Let~$v_i\in A$ and~$v_j\in B$.
If~$(v_i,v_j)\in E$, then the corresponding sticks have to intersect,
which they do if and only if $x_j-x_i\le \min\{\ell_i,\ell_j\}$.
If~$i<j$ and~$(v_i,v_j)\notin E$, then the corresponding sticks must not intersect,
so we require $x_j-x_i > \min\{\ell_i,\ell_j\}\ge \min\{\ell_i,\ell_j\}+\varepsilon$.
This easily gives a system of difference constraints with~$O(n^2)$ constraints.
We argue that a linear number suffices.

Let~$v_i\in A$. Let~$j$ be the largest~$j$ such that~$(v_i,v_j)\in E$ and~$\ell_j\ge \ell_i$.
We add a constraint~$x_j - x_i\le \ell_i$. 
Further, let~$k$ be the smallest~$k$ such that~$(v_i,v_k)\notin E$ and~$\ell_k\ge \ell_i$.
We add a constraint~$x_k-x_i> \ell_i \Leftrightarrow x_i - x_k\le -\ell_i-\varepsilon$. 
Symmetrically, let~$v_i\in B$. Let~$j$ be the smallest~$j$ such that~$(v_j,v_i)\in E$ and~$\ell_j> \ell_i$. We add a constraint~$x_i - x_j\le \ell_i$. 
Further, let~$k$ be the largest~$k$ such that~$(v_k,v_i)\notin E$ and~$\ell_k> \ell_i$.
We add a constraint~$x_i-x_k> \ell_i \Leftrightarrow x_k - x_i \le -\ell_i-\varepsilon$. 

We now argue that these constraints are sufficient to ensure that~$G$ is 
represented by a solution of the system.
Let~$v_i\in A$ and~$v_j\in B$.
If~$i>j$, then the corresponding sticks 
cannot intersect, which is ensured by the fixed order.
So assume that~$i<j$. If~$\ell_j\ge\ell_i$ and~$(v_i,v_j)\in E$, then we either
have the constraint~$x_j - x_i\le \ell_i$, or we have a constraint 
~$x_k - x_i\le \ell_i$ with $i<j<k$; together with the order constraints, this ensure
that $x_j-x_i\le x_k-x_i\le \ell_i$.
If~$\ell_j\ge\ell_i$ and~$(v_i,v_j)\notin E$, then we either have the constraint 
$x_i - x_j\le -\ell_i-\varepsilon$, or we have a constraint 
$x_i - x_k\le -\ell_i-\varepsilon$ with $i<k<j$; together with the order constraints, 
this ensure that $x_i-x_j\le x_i-x_k \le -\ell_i-\varepsilon$.
Symmetrically, the constraints are also sufficient for~$\ell_j<\ell_i$.
We obtain a system of difference constraints with~$n$ variables and at most $3n-1$ constraints proving Theorem~\ref{thm:stick-fixed-lp}.
By Lemma~\ref{lem:STICKAB-without-isolated-is-unique}, there is at most one realizable order of vertices for a \stickfixAB instance without isolated vertices, which can be found in linear time and proves~Corollary~\ref{cor:stickab-fixed-noisolated}.

\begin{thm}
	\label{thm:stick-fixed-lp}
	\stickfix can be solved in $O((|A|+|B|)^{2})$ time if we are given a total order of the vertices.
\end{thm}

\begin{cor}
	\label{cor:stickab-fixed-noisolated}
	\stickfixAB can be solved in $O((|A|+|B|)^{2})$ time if there are no isolated vertices.
\end{cor}

\section{Open Problems}
\label{sec:open_problems}

We have shown that \stickfix is NP-complete even if the sticks have only
three different lengths, while \stickfix for unit-length
sticks is solvable in linear time.  But what is the computational
complexity of \stickfix for sticks with one of \emph{two} lengths?
Also, the three different lengths in our proof depend on the number of sticks.
Is \stickfix still NP-complete if the fixed lengths are bounded?

We have shown that \stickfixAB is NP-complete if there are isolated vertices (in at least one of the bipartitions).
In our NP-hardness reduction we use a linear number of isolated vertices.
Clearly, \stickfixAB is in XP in the number~$n_\textrm{isolated}$ of isolated vertices.
An XP-algorithm could first compute the unique ordering of the non-isolated
vertices and then try to insert each isolated vertex at each possible position in the permutation brute-force.
However, the question remains open whether \stickfixAB is fixed-parameter tractable (FPT) in~$n_\textrm{isolated}$
\footnote{This question has been asked by Pawe\l{} Rz\k{a}\.{z}ewski at the 27th International Symposium on Graph Drawing and Network Visualization 2019 (GD'19) in Prague.}.

Additionally, the complexity of the original problem \stick,
i.e., recognizing a stick graph without vertex order or stick lengths,
is still open.


\section*{Acknowledgments}

We thank Anna Lubiw for detailed information about our most relevant
reference~\cite{dhklm-rdsg-TCS19} and for her constructive remarks
regarding an earlier version of this paper.



\bibliographystyle{abbrvurl}
\bibliography{abbrv,intersection-graphs}

\begin{thebibliography}{10}

\bibitem{BoothL76}
K.~S. Booth and G.~S. Lueker.
\newblock Testing for the consecutive ones property, interval graphs, and graph
  planarity using {PQ}-tree algorithms.
\newblock {\em J. Comput. Syst. Sci.}, 13(3):335--379, 1976.
\newblock \href {http://dx.doi.org/10.1016/S0022-0000(76)80045-1}
  {\path{doi:10.1016/S0022-0000(76)80045-1}}.

\bibitem{CabelloJ17}
S.~Cabello and M.~Jej\v{c}i\v{c}.
\newblock Refining the hierarchies of classes of geometric intersection graphs.
\newblock {\em Electr. J. Comb.}, 24(1):P1.33, 2017.
\newblock URL:
  \url{http://www.combinatorics.org/ojs/index.php/eljc/article/view/v24i1p33}.

\bibitem{CardinalFMTV-JGAA18}
J.~Cardinal, S.~Felsner, T.~Miltzow, C.~Tompkins, and B.~Vogtenhuber.
\newblock Intersection graphs of rays and grounded segments.
\newblock {\em J. Graph Algorithms Appl.}, 22(2):273--295, 2018.
\newblock \href {http://dx.doi.org/10.7155/jgaa.00470}
  {\path{doi:10.7155/jgaa.00470}}.

\bibitem{CatanzaroCFHHHS17}
D.~Catanzaro, S.~Chaplick, S.~Felsner, B.~V. Halld{\'{o}}rsson, M.~M.
  Halld{\'{o}}rsson, T.~Hixon, and J.~Stacho.
\newblock Max point-tolerance graphs.
\newblock {\em Discrete Appl. Math.}, 216:84--97, 2017.
\newblock \href {http://dx.doi.org/10.1016/j.dam.2015.08.019}
  {\path{doi:10.1016/j.dam.2015.08.019}}.

\bibitem{ChalopinG09}
J.~Chalopin and D.~Gon{\c{c}}alves.
\newblock Every planar graph is the intersection graph of segments in the
  plane: Extended abstract.
\newblock In {\em STOC}, pages 631--638. ACM, 2009.
\newblock \href {http://dx.doi.org/10.1145/1536414.1536500}
  {\path{doi:10.1145/1536414.1536500}}.

\bibitem{ChaplickDKMS14}
S.~Chaplick, P.~Dorbec, J.~Kratochv{\'{\i}}l, M.~Montassier, and J.~Stacho.
\newblock Contact representations of planar graphs: Extending a partial
  representation is hard.
\newblock In D.~Kratsch and I.~Todinca, editors, {\em WG}, volume 8747 of {\em
  LNCS}, pages 139--151. Springer, 2014.
\newblock \href {http://dx.doi.org/10.1007/978-3-319-12340-0\_12}
  {\path{doi:10.1007/978-3-319-12340-0\_12}}.

\bibitem{ChaplickFHW18}
S.~Chaplick, S.~Felsner, U.~Hoffmann, and V.~Wiechert.
\newblock Grid intersection graphs and order dimension.
\newblock {\em Order}, 35(2):363--391, 2018.
\newblock \href {http://dx.doi.org/10.1007/s11083-017-9437-0}
  {\path{doi:10.1007/s11083-017-9437-0}}.

\bibitem{ChaplickHOSU17}
S.~Chaplick, P.~Hell, Y.~Otachi, T.~Saitoh, and R.~Uehara.
\newblock Ferrers dimension of grid intersection graphs.
\newblock {\em Discrete Appl. Math.}, 216:130--135, 2017.
\newblock \href {http://dx.doi.org/10.1016/j.dam.2015.05.035}
  {\path{doi:10.1016/j.dam.2015.05.035}}.

\bibitem{dhklm-rdsg-TCS19}
F.~{De Luca}, M.~I. Hossain, S.~G. Kobourov, A.~Lubiw, and D.~Mondal.
\newblock Recognition and drawing of stick graphs.
\newblock {\em Theor. Comput. Sci.}, 796:22--33, 2019.
\newblock \href {http://dx.doi.org/10.1016/j.tcs.2019.08.018}
  {\path{doi:10.1016/j.tcs.2019.08.018}}.

\bibitem{FelsnerMM13}
S.~Felsner, G.~B. Mertzios, and I.~Mustata.
\newblock On the recognition of four-directional orthogonal ray graphs.
\newblock In K.~Chatterjee and J.~Sgall, editors, {\em {MFCS}}, volume 8087 of
  {\em {LNCS}}, pages 373--384. Springer, 2013.
\newblock Some results herein are incomplete, see the warning in the full
  version: \url{http://page.math.tu-berlin.de/~felsner/Paper/dorgs.pdf}.
\newblock \href {http://dx.doi.org/10.1007/978-3-642-40313-2\_34}
  {\path{doi:10.1007/978-3-642-40313-2\_34}}.

\bibitem{Garey1979}
M.~R. Garey and D.~S. Johnson.
\newblock {\em Computers and Intractability: {A} Guide to the Theory of
  {NP}-Completeness}.
\newblock W. H. Freeman, 1979.

\bibitem{HalldorssonATI11}
B.~V. Halld{\'{o}}rsson, D.~Aguiar, R.~Tarpine, and S.~Istrail.
\newblock The {C}lark phaseable sample size problem: Long-range phasing and
  loss of heterozygosity in {GWAS}.
\newblock {\em J. Comput. Biol.}, 18(3):323--333, 2011.
\newblock \href {http://dx.doi.org/10.1089/cmb.2010.0288}
  {\path{doi:10.1089/cmb.2010.0288}}.

\bibitem{HartmanNZ91}
I.~B. Hartman, I.~Newman, and R.~Ziv.
\newblock On grid intersection graphs.
\newblock {\em Discrete Math.}, 87(1):41--52, 1991.
\newblock \href {http://dx.doi.org/10.1016/0012-365X(91)90069-E}
  {\path{doi:10.1016/0012-365X(91)90069-E}}.

\bibitem{JungerLM98}
M.~J{\"{u}}nger, S.~Leipert, and P.~Mutzel.
\newblock Level planarity testing in linear time.
\newblock In S.~H. Whitesides, editor, {\em GD}, volume 1547 of {\em LNCS},
  pages 224--237. Springer, 1998.
\newblock \href {http://dx.doi.org/10.1007/3-540-37623-2\_17}
  {\path{doi:10.1007/3-540-37623-2\_17}}.

\bibitem{KlavikOS19}
P.~Klav{\'{\i}}k, Y.~Otachi, and J.~Sejnoha.
\newblock On the classes of interval graphs of limited nesting and count of
  lengths.
\newblock {\em Algorithmica}, 81(4):1490--1511, 2019.
\newblock \href {http://dx.doi.org/10.1007/s00453-018-0481-y}
  {\path{doi:10.1007/s00453-018-0481-y}}.

\bibitem{KoblerK015}
J.~K{\"{o}}bler, S.~Kuhnert, and O.~Watanabe.
\newblock Interval graph representation with given interval and intersection
  lengths.
\newblock {\em J. Discrete Algorithms}, 34:108--117, 2015.
\newblock \href {http://dx.doi.org/10.1016/j.jda.2015.05.011}
  {\path{doi:10.1016/j.jda.2015.05.011}}.

\bibitem{Kratochvil94}
J.~Kratochv{\'{\i}}l.
\newblock A special planar satisfiability problem and a consequence of its
  {NP}-completeness.
\newblock {\em Discrete Appl. Math.}, 52(3):233--252, 1994.
\newblock \href {http://dx.doi.org/10.1016/0166-218X(94)90143-0}
  {\path{doi:10.1016/0166-218X(94)90143-0}}.

\bibitem{KratochvilM94}
J.~Kratochv{\'{\i}}l and J.~Matou{\v{s}}ek.
\newblock Intersection graphs of segments.
\newblock {\em J. Comb. Theory, Series B}, 62(2):289--315, 1994.
\newblock \href {http://dx.doi.org/10.1006/jctb.1994.1071}
  {\path{doi:10.1006/jctb.1994.1071}}.

\bibitem{KratschMMS06}
D.~Kratsch, R.~M. McConnell, K.~Mehlhorn, and J.~P. Spinrad.
\newblock Certifying algorithms for recognizing interval graphs and permutation
  graphs.
\newblock {\em {SIAM} J. Comput.}, 36(2):326--353, 2006.
\newblock \href {http://dx.doi.org/10.1137/S0097539703437855}
  {\path{doi:10.1137/S0097539703437855}}.

\bibitem{l-tscr-DAM97}
W.~N. Li.
\newblock Two-segmented channel routing is strong {NP}-complete.
\newblock {\em Discrete Appl. Math.}, 78(1-3):291--298, 1997.
\newblock \href {http://dx.doi.org/10.1016/S0166-218X(97)00020-6}
  {\path{doi:10.1016/S0166-218X(97)00020-6}}.

\bibitem{Matousek14}
J.~Matou{\v{s}}ek.
\newblock Intersection graphs of segments and $\exists\mathbb{R}$.
\newblock ArXiv, \url{https://arxiv.org/abs/1406.2636}, 2014.

\bibitem{PeerS97}
I.~Pe'er and R.~Shamir.
\newblock Realizing interval graphs with size and distance constraints.
\newblock {\em {SIAM} J. Discrete Math.}, 10(4):662--687, 1997.
\newblock \href {http://dx.doi.org/10.1137/S0895480196306373}
  {\path{doi:10.1137/S0895480196306373}}.

\bibitem{Schaefer09}
M.~Schaefer.
\newblock Complexity of some geometric and topological problems.
\newblock In {\em GD}, volume 5849 of {\em LNCS}, pages 334--344. Springer,
  2009.
\newblock \href {http://dx.doi.org/10.1007/978-3-642-11805-0\_32}
  {\path{doi:10.1007/978-3-642-11805-0\_32}}.

\bibitem{SenS94}
M.~K. Sen and B.~K. Sanyal.
\newblock Indifference digraphs: {A} generalization of indifference graphs and
  semiorders.
\newblock {\em {SIAM} J. Discrete Math.}, 7(2):157--165, 1994.
\newblock \href {http://dx.doi.org/10.1137/S0895480190177145}
  {\path{doi:10.1137/S0895480190177145}}.

\bibitem{ShresthaTTU11}
A.~M.~S. Shrestha, A.~Takaoka, S.~Tayu, and S.~Ueno.
\newblock On two problems of nano-{PLA} design.
\newblock {\em {IEICE} Transactions}, 94-D(1):35--41, 2011.
\newblock \href {http://dx.doi.org/10.1587/transinf.E94.D.35}
  {\path{doi:10.1587/transinf.E94.D.35}}.

\bibitem{sbs-bpg-DAM87}
J.~Spinrad, A.~Brandst{\"a}dt, and L.~Stewart.
\newblock Bipartite permutation graphs.
\newblock {\em Discrete Appl. Math.}, 18(3):279--292, 1987.
\newblock \href {http://dx.doi.org/10.1016/S0166-218X(87)80003-3}
  {\path{doi:10.1016/S0166-218X(87)80003-3}}.

\end{thebibliography}

\end{document}